\documentclass[a4paper,onecolumn, 11pt, accepted=2023-10-25]{quantumarticle}
\pdfoutput=1

\usepackage{graphicx}
\usepackage{amsmath}
\usepackage{amssymb}
\usepackage{amsthm}
\usepackage{algpseudocode}
\usepackage{algorithm}
\usepackage{hyperref}
\usepackage{comment}
\usepackage{subcaption}
\usepackage{xcolor}

\newcommand{\be}{\begin{equation}}
\newcommand{\ee}{\end{equation}}
\newcommand{\ba}{\begin{array}}
\newcommand{\ea}{\end{array}}
\newcommand{\bea}{\begin{eqnarray}}
\newcommand{\eea}{\end{eqnarray}}

\newcommand{\ra}{\rangle}
\newcommand{\la}{\langle}

\newcommand{\calA}{{\cal A }}

\newcommand{\calS}{{\cal S }}

\newcommand{\EE}{\mathbb{E}}

\newcommand{\RR}{\mathbb{R}}

\newcommand{\od}{\mathsf{od}}
\newcommand{\ket}[1]{|#1\rangle}
\newcommand{\bra}[1]{\langle #1|}
\newcommand{\braket}[2]{\langle #1|#2\rangle}
\newcommand{\Exp}{\mathop\mathbb{E}}

\newcommand{\Var}{\mathop{\text{Var}}}
\newcommand{\cov}{\text{cov}}
\newcommand{\cor}{\text{cor}}

\newtheorem{lemma}{Lemma}
\newtheorem{fact}{Fact}

\newtheorem{theorem}{Theorem}

\newcommand{\footremember}[2]{%
    \footnote{#2}
    \newcounter{#1}
    \setcounter{#1}{\value{footnote}}%
}
\newcommand{\footrecall}[1]{%
    \footnotemark[\value{#1}]%
}
\begin{document}

\title{A rapidly mixing Markov chain from any gapped quantum many-body system}

\author{Sergey Bravyi}
\affiliation{IBM Quantum, IBM T.J. Watson Research Center, Yorktown Heights, USA}

\author{Giuseppe Carleo}
\affiliation{\'{E}cole Polytechnique F\'{e}d\'{e}rale de Lausanne (EPFL), Institute of Physics, CH-1015 Lausanne, Switzerland}

\author{David Gosset}
\affiliation{Department of Combinatorics and Optimization and Institute for Quantum Computing, University of Waterloo}
\affiliation{Perimeter Institute for Theoretical Physics, Waterloo, Canada}

\author{Yinchen Liu}
\affiliation{Department of Combinatorics and Optimization and Institute for Quantum Computing, University of Waterloo}
\affiliation{Perimeter Institute for Theoretical Physics, Waterloo, Canada}

\maketitle

\begin{abstract}

We consider the computational task of sampling a  bit string $x$ from a 
distribution $\pi(x)=|\la x|\psi\ra|^2$,  where $\psi$ is the unique ground state of a local Hamiltonian $H$. Our main result describes a direct link between the inverse spectral gap of $H$ and the mixing time of an associated continuous-time Markov Chain with steady state $\pi$. The Markov Chain can be implemented efficiently whenever 
ratios of ground state amplitudes $\la y|\psi\ra/\la x|\psi\ra$ are efficiently computable, the spectral gap of $H$ is at least inverse polynomial in the system size, 
 and the starting state of the chain satisfies a mild technical condition that can be efficiently checked. This extends a previously known relationship between sign-problem free Hamiltonians and Markov chains. The tool which enables this generalization is the so-called fixed-node Hamiltonian construction, previously used in 
Quantum Monte Carlo simulations to address the fermionic sign problem. 
We implement the proposed sampling algorithm numerically and use it to sample from the ground state of Haldane-Shastry Hamiltonian with up to $56$ qubits.
We observe empirically that our Markov chain based on the fixed-node Hamiltonian
mixes more rapidly than the standard Metropolis-Hastings Markov chain.
 \end{abstract}

The task of generating samples from a given probability distribution 
underlies almost all randomized algorithms used in computational physics, machine learning, and optimization.  In many applications the target distribution is \textit{efficiently computable} in the sense that there is a polynomial-time subroutine that computes the relative probabilities of any two given elements. A paradigmatic example is the (classical) Boltzmann distribution $\pi(x)\sim e^{-E(x)/T}$ associated with an efficiently computable energy function $E(x)$ and a temperature $T$. In this case the ratio of Boltzmann probabilities
 $\pi(y)/\pi(x)$ is simply related to the energy difference $E(y)-E(x)$. The fundamental obstacle in such cases is that \textit{distributions that are efficiently computable may still be challenging to sample from}. In particular, there are many examples when generating a sample from an efficiently computable distribution is an NP-hard problem.
This includes low-temperature Boltzmann distributions with energy function $E(x)$ that is given by a 3-SAT formula or an Ising spin glass~\cite{barahona1982computational} and certain distributions described by
neural networks of RBM type~\cite{long2010restricted}. 

It is therefore interesting to ask:

\vspace{0.3cm}
\noindent\textbf{Q1}:\textit{Which efficiently computable distributions can also be efficiently sampled?}
\vspace{0.3cm}

\noindent and, more broadly, 

\vspace{0.3cm}
\noindent\textbf{Q2}:\textit{Which distributions admit an efficient reduction from sampling to computing probabilities?}
\vspace{0.2cm}

Here we note that a polynomial-time reduction between the two tasks may exist for a broader family of distributions, some of which may not be efficiently computable.

A natural way to address these questions is to use Markov Chain Monte Carlo (MCMC), an empirically successful and versatile algorithmic tool for sampling probability distributions. MCMC methods work by constructing a Markov chain $M$ such that the target distribution $\pi$ is the unique steady distribution of $M$. A distribution $\pi_t$ generated after implementing  $t$ steps of the chain $M$ approximates
the steady distribution $\pi$ provided that the number of steps $t$ is large compared with the mixing time
of  $M$. A well known sampling algorithm in this category is the Metropolis-Hastings Markov chain \cite{hastings_monte_1970} and its variations. Although MCMC methods are widely used in practice, their main limitation is the difficulty of obtaining rigorous  upper bounds on the mixing time of Markov chains. 
Such bounds can be established only in certain special cases using techniques such as 
the canonical paths method, coupling of Markov chains, or the conductance bound, see e.g.~\cite{levin2017markov}.

In this work we are interested in variants of the questions \textbf{Q1} and \textbf{Q2} for probability distributions originating from ground states of quantum many-body systems. Can we design specialized MCMC sampling algorithms that exploit their structure? Are efficiently computable distributions \textit{that arise from ground states} efficiently samplable? 

In particular, we consider a system of $n$ qubits with few-qubit interactions described
by a {$k$-local} Hamiltonian 
$H=\sum_{a=1}^m H_a$,
where $H_a$ are Hermitian operators acting non-trivially on subsets of at most $k=O(1)$ qubits.
The Hamiltonian may or may not be local in the geometric sense.
For example, a $2$-local Hamiltonian can describe a chain of qubits with long-range two-qubit interactions. 
We choose the energy scale such that $\|H_a\|\le 1$ for all $a$.
Let $\psi$ be the ground state of $H$. We assume that the ground state of $H$
is non-degenerate and separated from excited states by an energy gap $\gamma>0$.
Our goal is to sample a bit string $x\in \{0,1\}^n$ from the ground state distribution 
\be
\label{pi(x)}
\pi(x)=|\la x|\psi\ra|^2.
\ee
It can be viewed as the zero-temperature quantum analogue of a classical Boltzmann distribution 
describing classical spins with few-spin interactions. 

To accomplish this goal one may construct a suitable quantum-to-classical mapping that converts
a quantum Hamiltonian with the ground state $\psi$ to a classical Markov chain $M$
with the steady distribution $\pi(x)=|\la x|\psi\ra|^2$. A possible choice for $M$ is a
Metropolis-Hastings (MH) Markov chain with local updates. In a simple, typical setting, each step of the MH chain flips a randomly chosen bit (or a subset of bits) of $x$ to propose a candidate next state $y$. 
The proposed state $y$ is accepted with the probability $\min\{1,\pi(y)/\pi(x)\}$
to ensure the detailed balance condition. As was shown in Ref.~\cite{bravyi2021simulate}, the mixing time
of the MH chain can be upper bounded as 
\begin{equation}
T\le \tilde{O}(n^k s/\gamma),
\label{eq:mixtime}
\end{equation}
where $\tilde{O}$ notation hides certain logarithmic factors and $s$ is a {\em sensitivity}
parameter defined as
$s=\max_{x\ne y} |\la y|H|x\ra \la y|\psi\ra/\la x|\psi\ra|$. For the widely studied family of \textit{sign-problem free}\footnote{A Hamiltonian $H$ with real matrix elements is sign-problem free if $\langle x|H|y\rangle \leq 0$ for all $x\neq y$.} Hamiltonians the sensitivity parameter can be bounded as \cite{bravyi2021simulate} 
\begin{equation}
s\leq 2\|H\|  \quad \qquad \text{(if $H$ is sign-problem free)}.
\label{eq:sens}
\end{equation}
For a local Hamiltonian of the type we consider we have $\|H\|=O(\mathrm{poly}(n))$ and so Eqs.~(\ref{eq:mixtime},\ref{eq:sens}) give a polynomial upper bound on the mixing time of the MH chain. Since each step of the chain makes use of a ratio $\pi(y)/\pi(x)$, this constitutes an efficient reduction from sampling to computing (ratios of) probabilities. In other words, we obtain a positive answer to $\textbf{Q2}$ for sign-problem free Hamiltonians \footnote{We note that a different efficient reduction from sampling to computing probabilities was previously known for this family of Hamiltonians \cite{aharonov2003adiabatic, bravyiterhal}.}. If we further specialize to frustration-free and sign-problem free Hamiltonians then the requisite ratios of probabilities can be computed efficiently~\cite{bravyiterhal} and we get a partial answer to $\textbf{Q1}$ as well. 

Unfortunately, for more general local Hamiltonians---those that may not be sign-problem free--- it is unknown whether the sensitivity $s$ admits a polynomial upper bound in terms of $n$ and $1/\gamma$, and significant differences may thwart this approach altogether. In the sign-problem free case the Perron Frobenius theorem implies that the ground state amplitudes $\langle x|\psi\rangle$ are nonnegative in the computational basis. For more general local Hamiltonians the entries of the ground state wavefunction may have nontrivial relative phases, but the MH chain 
ignores this information. Indeed, this chain only depends on the ratios of probabilities $\pi(y)/\pi(x)$. For these reasons it appears unlikely that the MH chain can be shown to  mix rapidly for general  local Hamiltonians with an inverse polynomial spectral gap. 

To proceed, we use a quantum-to-classical mapping
based on a proposal from Ref.~\cite{ten1995proof} which was originally introduced to circumvent the fermionic sign problem in Quantum Monte Carlo simulations \cite{foulkes2001quantum, becca2017quantum}.
This construction maps any $k$-local $n$-qubit Hamiltonian $H$ with real matrix elements,
unique ground state $\psi$, and a spectral gap $\gamma$
to a new $n$-qubit Hamiltonian $F$---defined in Eq.~\eqref{fixed_node_def}---such that: (i) $F$ and $H$ have the same ground state $\psi$,
(ii) the spectral gap of $F$ is at least $\gamma$, and (iii) the Hamiltonian $F$
is sign-problem free, modulo a simple basis change, see Section \ref{subs:markov_chain} for details. 
Following Ref.~\cite{ten1995proof}  we refer to $F$ as a {\em fixed-node} Hamiltonian.
Importantly, matrix elements of $F$ are efficiently computable given an efficient subroutine
for computing the ratio of amplitudes 
\begin{equation}
\qquad \qquad \qquad \qquad \qquad  x,y\longrightarrow \la y|\psi\ra/\la x|\psi\ra.   \qquad (\textbf{Amplitude computation subroutine}) \nonumber
\end{equation}
We are thus led to consider natural variants of the questions $\textbf{Q1},\textbf{Q2}$ where probabilities are replaced by amplitudes. This allows us to exploit the information encoded in the amplitudes' relative phase, which is an important feature of the ground state. We note that the model where a quantum state
can be accessed via the amplitude computation subroutine has been recently studied in~\cite{havlicek2022amplitude}.

Since the fixed node Hamiltonian is sign-problem free, we can sample the ground state distribution $\pi(x)$
using the MH chain and 
upper bound the mixing time of the chain using Eqs.~(\ref{eq:mixtime},\ref{eq:sens}). Unfortunately this is not useful because, as we shall see below, the norm of the fixed-node Hamiltonian can be unbounded. 

Instead, we introduce a quantum-to-classical mapping that yields a {\em continuous-time}  Markov chain rather than a regular
discrete-time chain. Recall that a continuous-time Markov chain (CTMC) with a state space $\calS$
defines a family of probability
distributions $\pi_t(x)$ where $t\ge 0$ is the evolution time
and $x\in \calS$ is the state reached at time $t$.
The time evolution of $\pi_t(x)$ is governed by a differential equation
\be
\label{CTMC1}
\frac{d}{dt} \pi_t(x) = \sum_{y\in \calS} \la x|G|y\ra \pi_t(y), 
\ee
where $G$ is a {\em generator matrix}. Rows and columns of $G$ are labeled by elements of $\calS$.
A matrix element $\la x|G|y\ra$ with $x\ne y$ can be viewed as the rate
of transitions from $y$ to $x$. Accordingly, all off-diagonal
elements of $G$ must be non-negative. The normalization  condition $\sum_{x\in \calS} \pi_t(x)=1$
is satisfied for all $t\ge 0$ as long as each column of $G$ sums to zero. 
A solution of Eq.~(\ref{CTMC1}) has the form 
$\pi_t(x)=\la x|e^{Gt} |x_{in}\ra$, 
where $x_{in}\in \calS$ is the starting state at time $t=0$ and $e^{Gt}$ denotes the matrix
exponential. 

Our CTMC is defined by a generator $G$ which is a suitably rescaled version of the fixed-node Hamiltonian $F$
associated with $H$. By design, it has a steady distribution $\pi(x)=\lim_{t\to \infty} \pi_t(x)$. It is given by
\be
\label{Gxy}
\la x|G|y\ra = \max\{0, - \la x|H|y\ra \la x|\psi\ra/\la y|\psi\ra\}
\ee
for $x\ne y$. Here and below we assume for simplicity that 
\begin{equation}
\pi(x)>0 \qquad \text{for all} \qquad x\in \{0,1\}^n.
\label{eq:pos}
\end{equation}
 Note that Eq.~\eqref{Gxy} also determines the diagonal matrix elements of $G$, since each column of $G$ sums to zero (due to the normalization condition).
Our main results are as follows. 
\begin{theorem}[\bf Rapid mixing]
Let  $H$ be a $k$-local  $n$-qubit Hamiltonian with 
real matrix elements in the standard basis,  unique ground state $\psi$,
and a spectral gap $\gamma$. 
Then a continuous-time Markov chain  with the state space $\{0,1\}^n$ and a
generator matrix $G$ defined in Eq.~(\ref{Gxy})
has a unique steady distribution $\pi(x)=|\la x|\psi\ra|^2$ and
obeys
\begin{equation}
\| \pi_t - \pi\|_1\le \frac{e^{-\gamma t}}{\sqrt{\pi(x_{in})}}
\label{eq:dist_t}
\end{equation}
for any $t\ge 0$ and any starting state $x_{in}\in \{0,1\}^n$.
Here
 $\pi_t(x)=\la x|e^{Gt} |x_{in}\ra$ is the distribution achieved by the Markov chain at a time $t$.
\label{thm:main1}
\end{theorem}
As we show below (Lemma \ref{lem:real}), the restriction to Hamiltonians $H$ with
real matrix elements is not essential and can be avoided by adding one ancillary qubit. 
Theorem \ref{thm:main1} shows that we may approximately sample from $\pi$ by running the continuous-time Markov chain for a total time $T\sim \gamma^{-1}\log(\pi(x_{in})^{-1})$. However it is not immediately clear how to simulate this process using resources polynomial in $T$, because of the significant caveat that the norm of $G$ may be large. This may lead to many transitions of the chain occuring in a very short time, and it prevents us from approximating the continuous-time chain by a discrete-time one obtained by naively discretizing the interval $[0,T]$.

Our saving grace is that we are able to establish a mild upper bound on the \textit{mean} number of transitions of the chain $G$ within a given interval, when the starting state is sampled from the steady distribution. This allows us to directly simulate the Markov chain of Theorem~\ref{thm:main1} using a truncated version of the well-known
Gillespie's algorithm~\cite{gillespie1977exact} in which we impose an upper limit on the total number of transitions of the chain. In this way we obtain the following
result.

\begin{theorem}[\bf Ground state sampling]
Let $\pi^*=\min_x \pi(x)$.
There exists a classical randomized algorithm that takes as input
a precision  $\epsilon>0$,
a starting state $x_{in} \in \{0,1\}^n$, makes at most 
\[
T=\frac{O( \epsilon^{-1} n^{2k} \| H\|)}{\gamma} \log{\left(\frac1{ \epsilon \pi^*}\right)}
\]
calls to the amplitude computation subroutine,
and either outputs a bit string $y\in \{0,1\}^n$
or declares an error. Let $S_{\epsilon}\subseteq \{0,1\}^n$  be the set of starting states $x_{in}$ for which the algorithm declares an error with probability at most $\epsilon/4$. The set $S_{\epsilon}$ is nonempty. Moreover, if the algorithm is run with starting state $x_{in}\in S_{\epsilon}$ and does not declare an error, then its output $y$ is sampled from a distribution $\epsilon$-close to $\pi$.
\label{thm:main}
\end{theorem}

The aforementioned caveat that $G$ has large matrix elements is also direcly related to the additional requirement above that we are provided with a good starting state $x_{in}$. Strictly speaking, Theorem \ref{thm:main} falls short of giving an efficient reduction from sampling to computing amplitudes of the ground state, since it requires this extra input. However, a good starting state $x_{in}$ can at least be verified using polynomial resources (and the amplitude computation subroutine): given $x\in \{0,1\}^n$ and $\epsilon>0$ we can decide with high probability whether or not $x\in S_{\epsilon}$ by running the above algorithm $O(\epsilon^{-2})$ times with starting state $x$, and using the results to compute an estimate of the probability that the algorithm declares an error.

In Section \ref{sec:qc} we describe the quantum-to-classical mapping in detail and we prove Theorems \ref{thm:main1} and \ref{thm:main}. Then in Section \ref{sec:hs} we demonstrate our algorithm for a concrete example (the Haldane-Shastry spin chain) where the amplitudes of the ground state are efficiently computable, and we compare our approach with the Metropolis-Hastings algorithm.

\section{Quantum-to-classical mapping\label{sec:qc}}
In this section we prove Theorems \ref{thm:main1} and \ref{thm:main}. 

In the following, for any matrix $M$ with real eigenvalues we write $\lambda_i(M)$ for the $i$-th smallest eigenvalue of $M$.  If $M$ is an $n$-qubit operator then $i=1,2,\ldots,2^n$. As described in the introduction, we shall consider an $n$-qubit, $k$-local Hamiltonian with a unique ground state $|\psi\rangle$ and we are interested in sampling from the distribution $\pi(x)=|\langle x|\psi\rangle|^2$. We shall assume Eq.~\eqref{eq:pos} holds, i.e., $|\langle x|\psi\rangle| >0$ for all $x\in \{0,1\}^n$.
\subsection{The continuous-time Markov chain}
\label{subs:markov_chain}

We first establish the following Lemma that shows we may restrict our attention to Hamiltonians with real matrix elements in the standard basis.

\begin{lemma}[\bf Reduction to real Hamiltonians]
Let $H$ be a local Hamiltonian with unique ground state $\psi$ and spectral gap $\gamma>0$, satisfying Eq.~\eqref{eq:pos}. There is a $O(n^k)$-sparse $(n+1)$-qubit Hamiltonian $H_{\mathrm{R}}$ with unique ground state  $|\phi\rangle=\mathrm{Re}(|\psi\rangle)|0\rangle+\mathrm{Im}(|\psi\rangle)|1\rangle$ and spectral gap at least $\mathrm{min}\{1,\gamma\}$. The $j$th nonzero entry of $H_\mathrm{R}$ in a given row can be computed using one call to the amplitude computation subroutine and efficient classical computation. 
\label{lem:real}
\end{lemma}
\begin{proof}
Suppose $H$ is an $n$-qubit, $k$-local Hamiltonian with a unique ground state $\psi$ such that $\langle x|\psi\rangle\neq 0$ for all $x$. Let us fix the global phase of $\psi$ such that $\langle 0^n|\psi\rangle\in \mathbb{R}$.

We may write
\[
H=A+i(B-B^{T})
\]
for real matrices $A,B$ which have at most $O(n^k)$ entries in each row. Now adjoin one ancilla qubit and consider the Hamiltonian
\[
H'=A\otimes I+(B-B^T)\otimes \left( \begin{array}{cc}
0 & -1 \\
1 &0 
\end{array} \right).
\] 
Note that $H'$ is a real symmetric matrix. Let us suppose that $|\phi\rangle$ is an eigenvector of $H$ with eigenvalue $\lambda$. Write $|\phi\rangle=|a\rangle+i|b\rangle$ where $a,b$ are real vectors, and consider the $n+1$-qubit states
\begin{align}
|\phi_0\rangle&=|a\rangle|0\rangle+|b\rangle|1\rangle\\
|\phi_1\rangle&=-|b\rangle|0\rangle+|a\rangle|1\rangle.
\end{align}
Observe that $|\phi_0\rangle$ is orthogonal to $|\phi_1\rangle$ and that each of these states is an eigenstate of $H'$ with eigenvalue $\lambda$. By letting $|\phi\rangle$ range over all $2^n$ eigenvectors of $H$ we obtain a complete set of $2\cdot 2^n$ eigenvectors for $H'$. In particular, the Hamiltonian $H'$ has the same spectral gap as $H$ but has a groundspace spanned by two eigenstates $\psi_0, \psi_1$.

In order to get rid of this degeneracy while preserving the spectral gap, we add a suitable positive semidefinite term to $H'$. 
For each $x\in \{0,1\}^n$ write the complex phase of $\langle x|\psi\rangle$ as
\[
\frac{\langle x|\psi\rangle}{|\langle x|\psi\rangle|}=\cos(\theta_x)+i\sin(\theta_x)
\]
where $\theta_x\in \mathbb{R}$. Then consider a Hamiltonian
\[
H''=H'+\sum_{x\in \{0,1\}^n} |x\rangle\langle x|\otimes |v_x\rangle \langle v_x| \quad \text{ where } \quad |v_x\rangle=-\sin(\theta_x)|0\rangle+\cos(\theta_x)|1\rangle.
\]
Note that $H''$ has $O(n^k)$ nonzero entries in each row and each one can be computed efficiently using one call to the amplitude computation subroutine\footnote{In particular, the angle $\theta_x$ can be obtained from the complex phase of the ratio $\langle x|\psi\rangle/\langle 0^n|\psi\rangle$.}.

Now observe that 
\[
\langle x|\langle v_x|\psi_0\rangle=-\sin(\theta_x)\mathrm{Re}(\langle x|\psi\rangle)+\cos(\theta_x)\mathrm{Im}(\langle x|\psi\rangle)=0
\]
where we used the definition of $\theta_x$. Therefore $H''|\psi_0\rangle=H'|\psi_0\rangle=\lambda_1(H)|\psi_0\rangle$. Moreover,

\[
\langle x|\langle v_x|\psi_1\rangle=+\sin(\theta_x)\mathrm{Im}(\langle x|\psi\rangle)+\cos(\theta_x)\mathrm{Re}(\langle x|\psi\rangle)=|\langle x|\psi\rangle|
\]
and therefore
\[
\left(\sum_{x\in \{0,1\}^n} |x\rangle\langle x|\otimes |v_x\rangle \langle v_x|\right)|\psi_1\rangle=\sum_{x} |x\rangle|v_x\rangle |\langle x|\psi\rangle|=\sum_{x}|x\rangle (\langle x|\otimes I)|\psi_1\rangle =|\psi_1\rangle
\]
From this we see that $\psi_1$ is an eigenvector of $H''$ with eigenvalue $\lambda_1(H)+1$. Since $H''\geq H'$, all other eigenvalues of $H''$ are at least $\lambda_2(H)$. Therefore $H''$ has unique ground state $\psi_0$ and its spectral gap is at least as large as 
\[
\min\{1, \lambda_2(H)-\lambda_1(H)\}=\min\{1,\gamma\}.
\]
\end{proof}

\begin{proof}[Proof of Theorem \ref{thm:main1}]
Our main technical tool is the so-called  effective  fixed-node  Hamiltonian
proposed by Ceperley et al~\cite{ten1995proof}.
It can be viewed as  a method of ``curing" the sign problem in 
Quantum Monte Carlo simulations. The method  is applicable whenever
amplitudes of the ground state can be efficiently computed.

As discussed above, in light of Lemma \ref{lem:real} we shall assume without loss of generality that $H$ has real matrix elements and a unique ground state with real entries in the standard basis. Define sets 
\begin{equation}
S^+ = \{ (x,y) \, : \, 
x\ne y \quad \mbox{and} \quad
 \la \psi |x\ra \la x|H|y\ra \la y|\psi\ra  > 0\},
\label{eq:fnode1}
\end{equation}
\begin{equation}
S^- = \{ (x,y) \, : \, 
x\ne y \quad \mbox{and} \quad \la \psi |x\ra\la x|H|y\ra \la y|\psi\ra  \le  0
\}.
\label{eq:fnode2}
\end{equation}
Here $x,y\in \{0,1\}^n$ are basis states.
Define a {\em fixed-node} Hamiltonian $F$
with matrix elements
\be
\label{fixed_node_def}
\la x |F|y\ra = \left\{ \ba{rcl}
0&\mbox{if}& (x,y)\in S^+,\\
\la x|H|y\ra &\mbox{if}& (x,y)\in S^-,\\
&&\\
\displaystyle
\la x|H|x\ra + \sum_{(x,z) \in S^+}  \la x|H|z\ra  \frac{ \la z|\psi\ra}{\la x|\psi\ra} &\mbox{if} & x=y.\\
\ea
\right.
\ee
Note that $F$ is stoquastic (sign problem free) 
modulo a change of basis $|x\ra \to \mathrm{sign}(\la x|\psi\ra)|x\ra$.
The following lemma is largely based on Ref.~\cite{ten1995proof}.
\begin{lemma}[\bf Fixed-node Hamiltonian]
\label{lemma:fixed_node}
The Hamiltonians $F$ and $H$ have the same unique ground state $\psi$ and 
the same ground energy. 
The spectral gap of $F$ is at least as large as the one of $H$.
\end{lemma}
\begin{proof}
First we claim that 
\be
\label{eq1}
F|\psi\ra = H|\psi\ra.
\ee
Indeed,
\[
\la x|F|\psi\ra =\left(\la x|H|x\ra +  \sum_{(x,z) \in S^+}  \la x|H|z\ra  \frac{ \la z|\psi\ra}{\la x|\psi\ra}\right)
\la x|\psi\ra+
 \sum_{(x,y)\in S^-} \la x|H|y\ra  \la y|\psi\ra=\la x|H|\psi\ra
\]
for any basis state $x$ which proves Eq.~(\ref{eq1}).
Next we claim  that 
\be
\label{eq2}
\la \phi|F|\phi\ra \ge \la \phi|H|\phi\ra \quad \mbox{for any state $\phi$}.
\ee
Indeed, 
\[
\la \phi | F - H|\phi\ra
=\sum_{(x,y)\in S^+}  \la x|H|y\ra \left( \frac{\la y|\psi\ra |\la x|\phi\ra|^2}{\la x|\psi\ra} -
 \la \phi|x\ra \la y|\phi\ra\right)
\]
Using the definition of $S^+$ one gets
\[
\la \phi | F- H|\phi\ra
=\sum_{(x,y)\in S^+}  |\la x|H|y\ra| \left( \frac{|\la y|\psi\ra| \cdot |\la x|\phi\ra|^2}{|\la x|\psi\ra|} -
  s(x,y) \la \phi|x\ra \la y|\phi\ra \right),
 \]
 where $s(x,y) =\mathrm{sign}(\la x|H|y\ra)$.
This is equivalent to
\[
\la \phi | F - H|\phi\ra=\frac12 \sum_{(x,y)\in S^+}  |\la x|H|y\ra| \left| \sqrt{\frac{|\la y|\psi\ra|}{|\la x|\psi\ra|}} \la x|\phi\ra - s(x,y)
  \sqrt{\frac{|\la x|\psi\ra|}{|\la y|\psi\ra|}} \la y|\phi\ra\right|^2.
\]
In particular, $\la \phi | F - H|\phi\ra\ge 0$.
We claim that the smallest and second-smallest eigenvalues of $H$ and $F$ satisfy
\be
\label{eq3}
\lambda_1(F)=\lambda_1(H) \quad \mbox{and} \quad \lambda_2(F)\ge \lambda_2(H).
\ee
Indeed,
from Eq.~(\ref{eq1}) one infers $\lambda_1(F)\le \la \psi|F|\psi\ra =
\la \psi|H|\psi\ra=\lambda_1(H)$. Suppose $\phi$ is a ground state of $F$. 
From Eq.~(\ref{eq2}) one gets $\lambda_1(F)=\la \phi|F|\phi\ra \ge \la \phi|H|\phi\ra \ge \lambda_1(H)$.
Thus $\lambda_1(H)=\lambda_1(F)$ and $\psi$ is a ground state of both $H$ and $F$.
Let $\phi$ be an eigenvector of $F$ orthogonal to $\psi$ such that $F|\phi\ra=\lambda_2(F)|\phi\ra$.
Since $H$ has the unique ground state $\psi$, one has
$\lambda_2(H)\le \la \phi|H|\phi\ra \le \la \phi|F|\phi\ra=\lambda_2(F)$.
Here the second inequality uses Eq.~(\ref{eq2}).
Thus $\lambda_2(F)\ge \lambda_2(H)$.
\end{proof}

The next step is to convert the fixed-node Hamiltonian $F$ to a Markov chain 
with the state space $\{0,1\}^n$ such that the ground state distribution $\pi$
is the unique steady state of the  chain.  One technical difficulty that
prevents us from applying the standard quantum-to-classical mapping
commonly used in Quantum Monte Carlo simulations
(see e.g. Section~8 of Ref.~\cite{aharonov2003adiabatic})
 is that the norm of $F$ may be unbounded.
Indeed,  a diagonal matrix element $\la x|F|x\ra$ depends on quantities
$\la x|H|z\ra \la z|\psi\ra/\la x|\psi\ra$ which can be arbitrarily large
even if the original Hamiltonian $H$  has a bounded norm.
Instead, we shall convert $F$ to a generator matrix describing a 
continuous-time Markov chain. We will see that the latter can be simulated efficiently using
the well-known Gillespie's algorithm~\cite{gillespie1977exact}. 
In the proof of Theorem \ref{thm:main} below we establish that the average number of iterations in Gillespie's algorithm 
can be bounded by a quantity that depends only on the spectral gap of $F$
and a certain ``off-diagonal norm" of $F$ which is at most $n^k \|H\|$
even if diagonal matrix elements of $F$ are unbounded.

Define a continuous-time Markov chain with the state space $\{0,1\}^n$
 as a family of stochastic
matrices of the form $e^{Gt}$, where we use the matrix exponential,
$t\ge 0$ is the evolution time, and $G$ is a   {\em generator matrix}
of size $2^n\times 2^n$.
By definition, the probability that the chain evolved for time $t$
makes a transition from a state $x$ to a state $y$ is given by 
$\la y|e^{Gt}|x\ra$. Equivalently, the probability of a transition
from $x$ to $y\in \{0,1\}^n\setminus \{x\}$ between time $t$ and $t+dt$ in the limit $dt\to 0$
is given by $\la y|G|x\ra dt$.
A valid generator matrix $G$ must have
non-negative off-diagonal matrix elements and each column of $G$
must sum to zero. This ensures that $e^{Gt}$ is a stochastic matrix
for any $t\ge 0$. A probability distribution $\eta$
is a steady state of the chain  iff
\be
\label{steady_state}
\sum_x \la y|G |x\ra \eta(x) = 0
\ee
for any  $y$. Equivalently, $\sum_x \la y|e^{Gt} |x\ra \eta(x)=\eta(y)$ for all $y$ and all $t\ge 0$.

Let $F$ be the fixed-node Hamiltonian constructed above.
Define a generator matrix $G$ such that 
\be
\label{G}
\la x|G|y\ra =  \lambda_1(F) \delta_{x,y} -  \la x|F|y\ra \frac{\la x|\psi\ra}{\la y|\psi\ra}
\ee
for all states $x,y$. Using Eq.~\eqref{fixed_node_def} we see that $G$ can equivalently be expressed as in Eq.~\eqref{Gxy}. Note that $\lambda_1(F)$ can be efficiently computed
by making $O(n^k)$ calls to the amplitude computation subroutine. 
Indeed, Lemma~\ref{lemma:fixed_node} implies that $\lambda_1(F)=\lambda_1(H)$.
Furthermore, 
\[
\lambda_1(H) = \frac{\la x|H|\psi\ra}{\la x|\psi\ra} = \sum_y \la x|H|y\ra\cdot  \frac{\la y|\psi\ra}{\la x|\psi\ra}
\]
for any basis state $x$. It remains to note that 
each row of $H$ has at most $O(n^k)$ non-zeros.

The following lemma implies that $G$ is a valid generator matrix
for a continuous-time Markov chain with the unique steady state $\pi$.
\begin{lemma}[\bf Generator matrix]
\label{lemma:gen}
The generator matrix $G$ has real non-positive eigenvalues.
Its  largest and second-largest eigenvalues are $0$ and $-\gamma_F$
respectively, where $\gamma_F$ is the spectral gap of the fixed-node
Hamiltonian $F$. The matrix exponential $e^{Gt}$
is a stochastic matrix for any $t\ge 0$. The distribution $\pi$ is the unique
steady state of $e^{Gt}$.
\end{lemma}
\begin{proof}
Let $D$ be a diagonal matrix such that $\la x|D|x\ra = \la x|\psi\ra$.
By definition, $G=D(\lambda_1(F) I - F)D^{-1}$.
Thus eigenvalues of $G$ coincide with eigenvalues of $\lambda_1(F) I - F$.
This proves the first and the second claims of the lemma. 

Let us check that $e^{Gt}$ is a stochastic matrix. 
It follows directly from the definitions that 
$G$ has non-negative off-diagonal elements.
Thus  all matrix elements of $e^{Gt}$ are non-negative.
Any column of $G$ sums to zero since
\[
\sum_x \la x|G|y\ra = \lambda_1(F) - \frac{\la \psi|F|y\ra}{\la y|\psi\ra}=
\lambda_1(F) - \frac{\lambda_1(F) \la \psi |y\ra}{\la \psi |y\ra} = 0.
\]
Here we noted that $\la \psi | F=\lambda_1(F)\la \psi|$ due to Lemma~\ref{lemma:fixed_node}.
This implies that any column of $e^{Gt}$ sums to one,
that is, $e^{Gt}$ is a stochastic matrix. 
Finally,
\begin{align*}
\sum_x \la y|G|x\ra \pi(x) &= \sum_x \la y| D(\lambda_1(F)I -F)D^{-1} |x\ra \pi(x)
 = \la y| \psi\ra \cdot \la y | \lambda_1(F)I -F|\psi\ra = 0
\end{align*}
for any state $y$. Thus $\pi$ is a steady state. This is a unique
steady state since the zero eigenvalue of $G$ is non-degenerate. 
\end{proof}

To complete the proof of Theorem \ref{thm:main1} we now establish the bound Eq.~\eqref{eq:dist_t}. 

Our proof strategy closely follows Ref.~\cite{diaconis1991geometric}, see Proposition~3 thereof.
First we claim the Markov chain $e^{Gt}$ obeys the detailed balance condition
\be
\label{db}
\pi(z)\la y|e^{Gt} |z\ra = \pi(y) \la z|e^{Gt}|y\ra
\ee
for all states $y,z$. Indeed, let $D$ be a diagonal matrix such that $\la z|D|z\ra = \la z|\psi\ra$
for all $z$.
We have 
\[
G=DMD^{-1} \quad \mbox{where} \quad M:=\lambda_1(F) I -F.
\]
Note that $M$ is a symmetric matrix since the fixed-node Hamiltonian $F$ is symmetric.
The identity $e^{Gt} = D e^{Mt} D^{-1}$ gives
\[
\pi(z)\la y|e^{Gt} |z\ra =\la y|\psi\ra  \la z|\psi\ra \la y|e^{Mt}|z\ra.
\]
Clearly this expression is symmetric under the exchange of $y$ and $z$, which
proves Eq.~(\ref{db}).
Using Cauchy-Schwartz one gets
\be
\label{mixing_eq1}
\| \pi_t - \pi\|_1^2 =\left( \sum_y \sqrt{\frac{\pi(y)}{\pi(y)}} |\pi_t(y) - \pi(y)|\right)^2
\le \sum_y \frac1{\pi(y)} |\pi_t(y)-\pi(y)|^2
= -1+\sum_y \frac{(\pi_t(y))^2}{\pi(y)}.
\ee
The detailed balance condition Eq.~(\ref{db})  gives
\be
\label{mixing_eq2}
 \frac{(\pi_t(y))^2}{\pi(y)} = \frac{\la y|e^{Gt}|x\ra^2}{\pi(y)} = \frac{\la x|e^{Gt}|y\ra\la y|e^{Gt} |x\ra}{\pi(x)}
 \ee
 for all $y$. Combining Eqs.~(\ref{mixing_eq1},\ref{mixing_eq2}) gives
\be
\label{mixing_eq3}
\| \pi_t - \pi\|_1^2 \le -1+ \frac{\la x|e^{2Gt}|x\ra}{\pi(x)} = -1+ \frac{\la x|e^{2Mt}|x\ra}{\pi(x)}.
\ee
Here we used an identity $e^{2Gt} = D e^{2Mt} D^{-1}$.
Consider an eigenvalue decomposition
\[
M = \sum_{i=1}^{2^n} \lambda_i(M) |\phi_i\ra\la \phi_i|
\]
such that $\lambda_i(M)$ is the $i$-th smallest eigenvalue of $M$
and $\phi_i$ is the corresponding eigenvector.
We have $\la \phi_i|\phi_j\ra=\delta_{i,j}$.
Since the matrices $M$ and $G$ are related by a similarity transformation,
Lemma~\ref{lemma:gen} implies that 
$\lambda_{2^n}(M)=0$ is the largest eigenvalue of $M$ 
 and $\lambda_{i}(M)\le -\gamma_F$
for all $i<2^n$. Furthermore, Lemma~\ref{lemma:fixed_node} implies $M|\psi\ra=0$, that is,
the only zero eigenvector of $M$ is $\phi_{2^n}=\psi$.
Thus we can write
\[
e^{2Mt} = |\psi\ra\la \psi|  + R, \quad \mbox{where} \quad R:=\sum_{i=1}^{2^n-1} e^{2 \lambda_i(M)t} |\phi_i\ra\la \phi_i|
\]
obeys  $\|R\|\le e^{-2\gamma_Ft}$.
Substituting $e^{2Mt} = |\psi\ra\la \psi|  + R$ into 
Eq.~(\ref{mixing_eq3}) gives
\be
\| \pi_t - \pi\|_1^2 \le -1 + 1+ \frac{\la x|R |x\ra}{\pi(x)} \le \frac{\|R\|}{\pi(x)}\le
 \frac{e^{-2\gamma_F t}}{\pi(x)}.
\label{eq:gfbound}
\ee
It remains to note that $\gamma_F\ge \gamma$, see Lemma~\ref{lemma:fixed_node}.
\end{proof}

\subsection{Sampling algorithm}

Let $x_{in}\in \{0,1\}^n$ be the starting state. 
We would like to sample $x$ from a distribution 
\[
\pi_t(x):=\la x|e^{Gt} |x_{in}\ra.
\]
To this end we use Gillespie's algorithm~\cite{gillespie1977exact}.
The algorithm takes as input the starting state $x_{in}$, evolution time $t\ge 0$,
and returns a sample $x$ from  the distribution $\pi_t(x)$.

\begin{center}
\begin{algorithm}
\caption{Gillespie's algorithm}
\begin{algorithmic}[1]
\State{$x\gets x_{in}$}
\State{$\tau \gets 0$} 
\State{$\xi(0)\gets x_{in}$}
\State{Sample $u \in [0,1]$ from the uniform distribution }
\State{$\Delta \tau \gets  \frac{\log{(1/u)}}{|\la x|G|x\ra|}$}
\State{Set $\xi(s)=x$ for all $s\in (\tau,\tau+\Delta \tau]$} 
\State{$\tau \gets \tau +\Delta \tau$}
\If{$\tau\ge t$}
\State{\textbf{return} $x$}
\EndIf
\State{Sample $y\in \{0,1\}^n\setminus \{x\}$ from the probability distribution $\frac{\la y|G|x\ra}{|\la x|G|x\ra|}$}
\State{$x\gets y$}
\State{\textbf{go to} line~4}
		\end{algorithmic}
\end{algorithm}
\end{center}

Lines~3 and 6 can be safely ignored as far as the implementation is concerned.
However it helps us to prove certain properties of the algorithm.
Namely, each run of the algorithm generates a 
continuous-time random walk on the set of bit strings
described by the piecewise constant function $\xi(s)$ with $s\in [0,t]$.
\begin{fact}[\textbf{Output of Gillespie's algorithm}]
For each $s\in [0,t]$, the random variable $\xi(s)\in \{0,1\}^n$ generated by Gillespie's algorithm is distributed according to $\pi_s(x)$.
\end{fact}
For completeness, we include a proof sketch for the correctness of Gillespie's algorithm.

\begin{proof}[Proof sketch]
The following derivations are based on Ref.~\cite{takahara2017notes}.  Let $P_{y,x}(\tau)=\Pr(\xi(\tau)=y|\xi(0)=x)=\bra{y}P(\tau)\ket{x}$ denote the probability of being in state $y$ at time $\tau$ given that the initial state is $\xi(0)=x$. We claim that $P(\tau)$ satisfies the differential equation $P'(\tau)=GP(\tau)$. Let $h>0$ where we think of $h$ as a small quantity tending to zero. Let $m(x,h)$ be a random variable that is equal to the number of transitions that occur in the time interval $[0,h]$ starting with initial state $x$.  Since the time $\Delta\tau$ at which the first transition occurs is exponentially distributed with rate $|\bra{x}G\ket{x}|$,
\begin{equation}
\Pr(m(x,h)=0)=e^{-|\bra{x}G\ket{x}|h}=1-|\bra{x}G\ket{x}|h+O(h^2).
\label{eq:zerom}
\end{equation}
Next we note that 
\[
\Pr(m(x,h)\geq 2)=\sum_{y\neq x}\Pr\left(\Delta\tau_x+\Delta\tau_y\leq h\right)\frac{\bra{y}G\ket{x}}{|\bra{x}G\ket{x}|},
\]
where $\Delta\tau_{x}, \Delta\tau_y$ are exponentially distributed  random variables with rates $|\bra{x}G\ket{x}|, |\bra{y}G\ket{y}|$ respectively. 
We have 
\[
\Pr\left(\Delta\tau_x+\Delta\tau_y\leq h\right) = 
\int_0^h dh_1 \int_0^{h-h_1} dh_2 |\la x|G|x\ra|e^{- |\la x|G|x\ra| h_1}
 |\la y|G|y\ra|e^{- |\la y|G|y\ra| h_2} = O(h^2).
\]
Thus $\Pr(m(x,h)\geq 2)=O(h^2)$ and
\begin{equation}
\Pr(m(x,h)=1)=|\bra{x}G\ket{x}|h+O(h^2).
\label{eq:onem}
\end{equation}
Eqs.~(\ref{eq:zerom},\ref{eq:onem}) imply that for every $x\in\{0,1\}^n$ and every $y\neq x$,
\[
P_{x,x}(h)=1-|\bra{x}G\ket{x}|h+O(h^2) \qquad \text{and} \qquad P_{x,y}(h)=\bra{y}G\ket{x}h+O(h^2).
\]
For every $\tau\geq 0$ and $x,y\in\{0,1\}^n$, we arrive at
\begin{align*}
P_{y,x}(\tau+h)&=P_{y,x}(\tau)P_{x,x}(h)+\sum_{z\neq x}P_{y,z}(\tau)P_{z,x}(h)\\
&=P_{y,x}(\tau)-|\bra{x}G\ket{x}|P_{y,x}(\tau)h+\sum_{z\neq x}P_{y,z}(\tau)\bra{z}G\ket{x}h+O(h^2).
\end{align*}
Hence, we conclude that
$$\bra{y}P'(\tau)\ket{x}=\lim_{h\rightarrow 0}\frac{P_{y,x}(\tau+h)-P_{y,x}(\tau)}{h}=\sum_{z\in\{0,1\}^n}\bra{y}P(\tau)\ket{z}\bra{z}G\ket{x}=\bra{y}P(\tau)G\ket{x}.$$
Solving this differential equation known as Kolmogorov's forward equation, we obtain the solution $P(s)=e^{Gs}$. This establishes the desired equivalence $\Pr(\xi(s)=x|\xi(0)=x_{in})=\bra{x}e^{Gs}\ket{x_{in}}=\pi_s(x)$.
\end{proof}

Let $m(x_{in}, t)$ be the number of {\em flips} performed by the algorithm,
that is, the number of times the function $\xi(s)$ changes its value.
One can easily check that the total number of calls
to the amplitude computation subroutine made at lines 5 and 11
 is at most 
$O(n^k)\cdot m(x_{in},t)$.
Thus the algorithm is efficient as long as the number of flips is not too large.
Our key observation is the following. 
\begin{lemma}[\bf Average number of flips]
Let $F_\od$ be  the off-diagonal part of the fixed-node Hamiltonian 
obtained from $F$ by setting to zero all diagonal matrix elements. Then 
\be
\label{flips}
\sum_{x_{in}} \pi(x_{in}) m(x_{in},t) = - t \la \psi |F_\od |\psi\ra.
\ee
\end{lemma}
We note that $\la \psi|F_\od |\psi\ra\le 0$ by definition of $F$, see Eq.~(\ref{fixed_node_def}).

 Eq.~\eqref{flips} is closely related to a known estimator for the off-diagonal part of the Hamiltonian in many-body simulations based on the continuous-time worldline quantum Monte Carlo method \cite{prokof1998exact} (see, e.g., Eq. (19) of \cite{smallgaps}).
\begin{proof}
Let $\rho(\xi)$ be the probability distributions over 
paths $\xi\, : \, [0,t]\to \{0,1\}^n$ generated by the algorithm
with the starting state $x_{in}$ sampled from $\pi(x_{in})$.
Since $\pi$ is the steady distribution of the considered Markov chain,
one infers that 
Gillespie's algorithm terminated at any intermediate time $s$ samples
the steady distribution, that is,
\be
\label{marginal_xi}
\mathrm{Pr}_{\xi \sim \rho}[ \xi(s)=x] = \pi(x)
\ee
for any fixed $s\in [0,t]$.
Now let us consider the probability $P_{x,s}(\xi)$ for the path $\xi$
with a fixed value $\xi(s)=x$ to 
contain a flip between time $s$ and $s + ds$
in the limit $ds \to 0$. This probability is
$
P_{x,s}(\xi)ds  =  
\sum_{y\ne x} \la y|G|x\ra ds$.
Thus the average number of flips between time $s$ and $s + ds$
is given by
\be
\label{flip1}
P_s ds := 
\EE_{\xi \sim \rho}  P_{\xi(s),s} ds = \sum_{y\ne x} \la y|G|x\ra \pi(x)  ds
= - \la \psi |F_\od |\psi \ra   ds.
\ee
Here we used Eq.~(\ref{marginal_xi}) and the definitions of $F$ and $G$, see Eqs.~(\ref{fixed_node_def},\ref{G}).
Integrating Eq.~(\ref{flip1}) over $s\in [0,t]$ completes the proof.
\end{proof}
By definition of $F$, see Eq.~(\ref{fixed_node_def}), one has 
\begin{align}
-\la \psi |F_\od|\psi\ra & =\sum_{(x,y)\in S^-} |\la \psi |x\ra \la x|H|y\ra \la y|\psi\ra |
\le \sum_{x,y} |\la \psi|x\ra\la x|H|y\ra\la y| \psi\ra| \nonumber \\
&
\leq\frac{1}{2}\sum_{x,y}|\bra{x}H\ket{y}|(|\braket{x}{\psi}|^2+|\braket{y}{\psi}|^2) \nonumber \\
&=\sum_{x}|\braket{x}{\psi}|^2\sum_{y}|\bra{x}H\ket{y}|
\leq
\max_x \sum_y 
|\bra{x}H\ket{y}|
\leq n^k\lVert H\rVert.
\label{norm_od}
\end{align}
Here we noted that each row of $H$ contains at most $n^k$ non-zero matrix
elements and each matrix element has magnitude at most $\|H\|$.
Thus the average number of flips performed by Gillespie's algorithm
is at most $tn^k \|H\|$, assuming that the starting state $x_{in}$
is sampled from the steady distribution $\pi(x_{in})$.

\begin{proof}[Proof of Theorem \ref{thm:main}]
We shall use a truncated version of  Gillespie's algorithm
which terminates whenever the condition $\tau \ge t$  at line~8 is satisfied or the number of flips exceeds a cutoff value $4\epsilon^{-1}tn^k \|H\|$. 
In the latter case the truncated algorithm declares an error.
By Markov's inequality, there exists at least one starting
state $x^{\star}_{in}$ such that the algorithm errs with the probability at most $\epsilon/4$. The truncated Gillespie's algorithm makes at most $O(t\epsilon^{-1}n^{2k} \|H\|)$
calls to the amplitude computation subroutine
since each flip requires $O(n^k)$ amplitude computations.
Let us write $\tilde{\pi}_{t}(x)$ for the probability distribution sampled by the truncated algorithm, with some starting state $x_{in}$, and conditioned on the event that no error occurs. Let $E$ be the event $m(x_{in},t)\geq 4\epsilon^{-1}tn^k \|H\|$ and let $P(E,x_{in},t)$ be the probability of this event, i.e., the probability that the truncated Gillespie algorithm outputs an error. Then
\[
\pi_t(x)=\tilde{\pi}_t(x)(1-P(E,x_{in},t))+\eta_t(x)P(E,x_{in},t)
\]
where $\eta_t(x)$ is the probability that Gillespie's algorithm outputs $x$ conditioned on event $E$. From the above we get 
\begin{align*}
    \|\tilde{\pi}_{t}-\pi_t\|_1\leq\sum_{x}\left|\eta_t(x)-\tilde{\pi}_t(x)\right|\cdot P(E,x_{in},t)\leq 2P(E,x_{in},t).
\end{align*}
For the special starting state $x^{\star}_{in}$ we see that the distribution sampled by the truncated Gillespie algorithm is $\epsilon/2$-close to the distribution $\pi_t$ sampled by the Gillespie algorithm without truncation.

Finally, we use Theorem \ref{thm:main1} to choose the evolution time $t$ in Gillespie's algorithm large enough such that the
distribution $\pi_t(x)$ is $\epsilon/2$-close to the steady distribution $\pi(x)$ in total variation distance. In particular, we conclude that the distributions $\pi_t(x)=\la x|e^{Gt}|x_{in}\ra$ and $\pi(x)$ are $\epsilon/2$-close in the total 
variation distance for all $t\ge t_{\epsilon/2}$, where
\[
t_\epsilon = \frac1{\gamma} \log{\left( \frac1{\epsilon \sqrt{\pi(x_{in})}}\right)}\le  \frac1{\gamma} \log{\left( \frac1{\epsilon \sqrt{\pi^*}}\right)}
\]
This concludes the proof of Theorem \ref{thm:main}.
\end{proof}

\section{Application to the Haldane-Shastry model\label{sec:hs}}

In this section we further investigate the properties of the continuous-time Markov Chain (CTMC) based on the fixed node Hamiltonian. Here we focus on a specific example of a quantum spin system where the amplitudes of the ground state can be computed exactly and efficiently. We implement the CTMC in software and use it to approximately sample from the resulting probability distribution.

The example we consider is a system of $L\geq 2$ qubits interacting according to the Haldane-Shastry Hamiltonian

\begin{equation}
H=\sum_{1\leq i<j\leq L}\frac{X_iX_j+Y_iY_j+Z_iZ_j}{4\left(\frac{L}{\pi}\sin\frac{\pi(i-j)}{L}\right)^2}
\label{eq:hsh}
\end{equation}
as defined in, e.g., Refs.~ \cite{stephan2017full} and \cite{pai2020disordered}. This Hamiltonian describes a spin chain with periodic boundary conditions (i.e., a ring) and long-range two-qubit interactions that depend on the shortest distance between the qubits on the ring. The Hamiltonian Eq.~\eqref{eq:hsh} has a unique ground state $\psi$ with the following properties \cite{stephan2017full, pai2020disordered}. For every $x\in\{0,1\}^L$ such that $|x|=\frac{L}{2}$,
\begin{equation}
\braket{x}{\psi}\propto \prod_{k=1}^L (-1)^{(k-1)x_k}\prod_{1\leq i<j\leq L}\left(\sin\frac{\pi(i-j)}{L}\right)^{2x_ix_j}
\label{eq:hsground}
\end{equation}
where $\prod_{k=1}^L (-1)^{(k-1)x_k}$ determines the relative sign, and for every $x\in\{0,1\}^L$ such that $|x|\neq\frac{L}{2}$, $\braket{x}{\psi}=0$. It is also known \cite{stephan2017full} that the spectral gap of $H$ scales inversely  with  system size $L$, and in particular satisfies
\begin{equation}
\gamma=\frac{2\pi c}{L}
\label{eq:gamcl}
\end{equation}
for some constant $c$. 

We chose this example because the entries of the ground state Eq.~\eqref{eq:hsground}
are efficiently computable, and because we are not aware of alternative algorithms with a provable polynomial runtime for sampling from the probability distribution 
$\pi(x)=|\la x|\psi\ra|^2$.
For example, in Appendix \ref{sec:lack} we show that the state $\psi$ is not a fermionic Gaussian state and therefore cannot be directly sampled using a naive application of free-fermion based methods. One can easily check that 
Haldane-Shastry Hamiltonian $H$
has a sign problem (that is, $H$ is not stoquastic). 
Indeed, each two-qubit term  in $H$ has a positive matrix element between basis
states $10$ and $01$. 
Moreover, in Appendix~\ref{sec:lack1} we use 
techniques of Ref.~\cite{klassen2019two}
to show that the sign problem in $H$ cannot be ``cured" by a local change of basis, that is,
a Hamiltonian $U^\dag HU$ has a sign problem for any product unitary $U=U_1\otimes U_2\otimes \cdots \otimes U_L$.
This suggests that standard Quantum Monte Carlo methods applicable to sign problem free  Hamiltonians
cannot be directly used to sample the ground state of $H$. 
It should be noted that Haldane-Shastry Hamiltonian admits a 2D generalization
with an efficiently computable ground state amplitudes~\cite{nielsen2012laughlin}.
We expect that our sampling algorithm can be applied to this generalized model as well.

The first step in our method for sampling from the probability distribution $\pi(x)$  is to construct the associated fixed-node Hamiltonian defined by Eq.~\eqref{fixed_node_def}. Note that all matrix elements of $H$ and all entries of $\ket{\psi}$ are real numbers, so in this example there is no need to preprocess the Hamiltonian using Lemma \ref{lem:real}. The associated CTMC is generated by the matrix $G$ from Eq.~\eqref{G}. The mixing time of this CTMC scales inversely with the spectral gap  $\gamma_F$ of $F$ (see Eq.~\eqref{eq:gfbound}). We have shown that $\gamma_F\geq \gamma$ holds in general, which, combined with Eq.~\eqref{eq:gamcl} implies that the mixing time is upper bounded as $O(L)$. Figure \ref{fig:gap} shows the spectral gaps $\gamma$ and $\gamma_F$ for the Haldane-Shastry Hamiltonian and small values of $L$, computed using exact numerical diagonalization. We find empirically that $\gamma_F$ has a milder scaling with system size $L$ than $\gamma$ and we therefore expect the CTMC to mix more rapidly than the rigorous bounds suggest. 

\begin{figure}[t!]
\begin{subfigure}[h]{0.45\textwidth}
\includegraphics[scale=0.55]{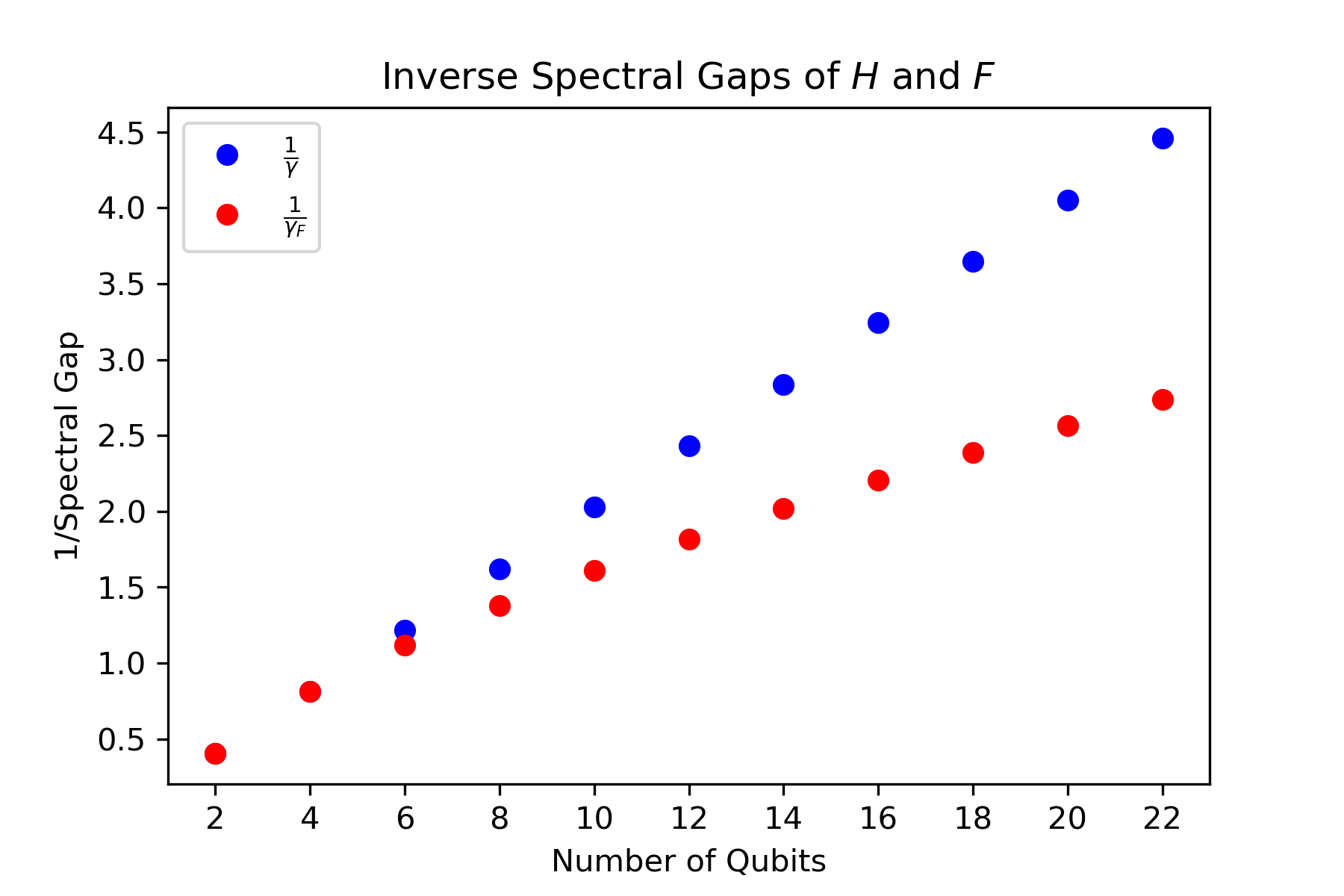}
\caption{}
\end{subfigure}
\hspace{1cm}
\begin{subfigure}[h]{0.45\textwidth}
\includegraphics[scale=0.55]{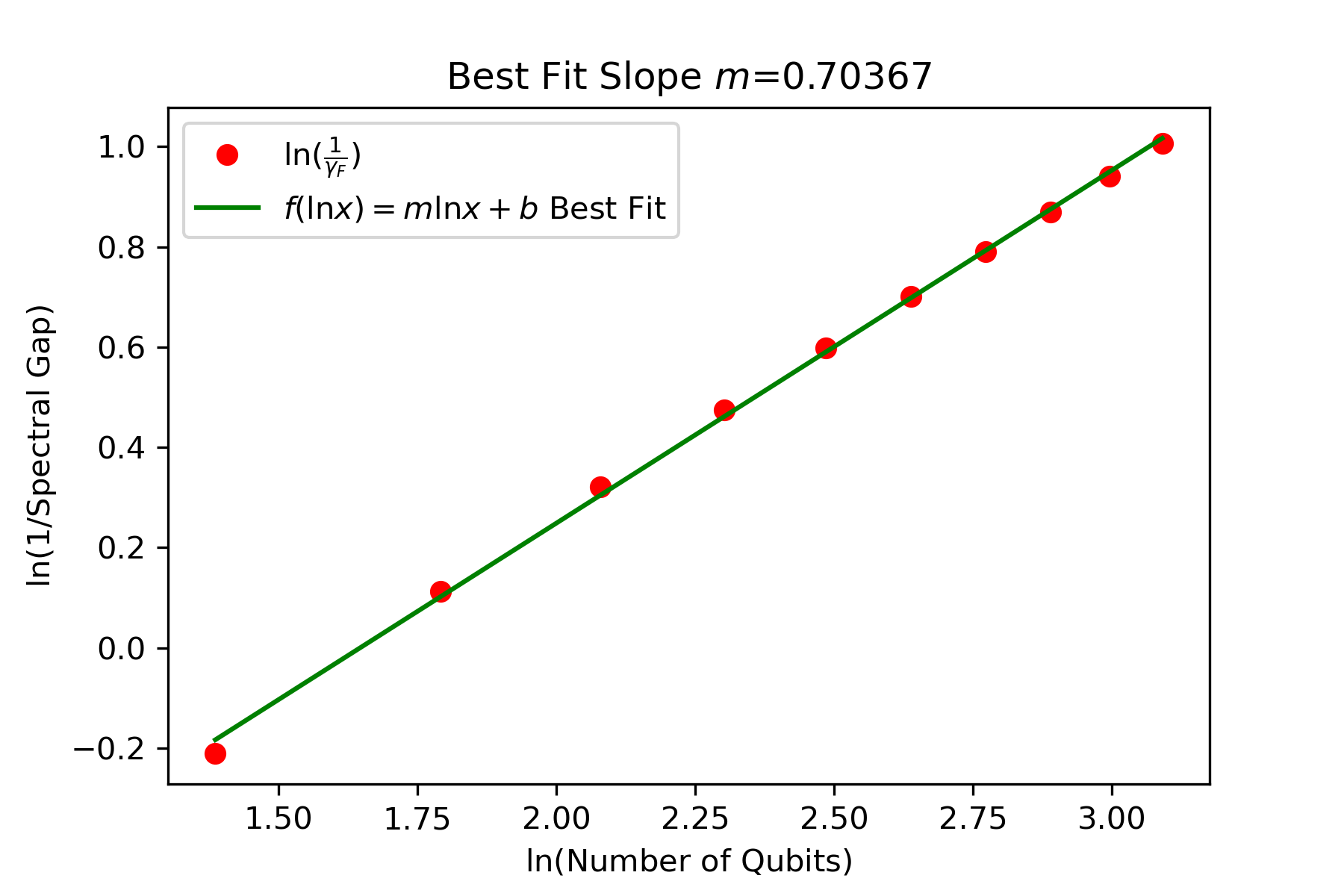}
\caption{}
\end{subfigure}
\caption{(a) The inverse of the spectral gap of the fixed-node Hamiltonian $F$ increases more slowly than that of the Haldane-Shastry Hamiltonian for small values of the number of qubits $L$ (computed using exact numerical diagonalization). (b) A linear fit for $L\geq 4$ on a log-log plot suggests a scaling of $\gamma_F\sim L^{-0.7\ldots}$.\label{fig:gap}}
\end{figure}

We implement Gillespie's algorithm using the fixed-node construction to approximately sample from the distribution $\pi(x)=|\braket{x}{\psi}|^2$, $x\in\{0,1\}^L$. The starting state is chosen randomly from $\text{Supp}(\ket{\psi})=\{x\in\{0,1\}^L:\braket{x}{\psi}\neq 0\}$. By our rigorous mixing time upper bound and numerically determined spectral gaps, the Markov chain is expected to converge  rapidly. To assess convergence numerically in practice, we use the vanilla $\hat{R}$ statistic discussed in \cite{vehtari2021rank} to select an appropriate ``burn-in'' period for our chain.

After allowing the CTMC to (approximately) converge, we are then able to estimate physical quantities in the ground state $\psi$. These can be compared with exact formulas that are available for the Haldane-Shastry ground state. For example, it is known from \cite{stephan2017full} that for every $1\leq i<j\leq L$, the two-point $ZZ$ correlation admits the formula
\begin{equation}
\bra{\psi}Z_iZ_j\ket{\psi}=\left(\sum_{k=1}^{L/2}\frac{\sin\frac{(2k-1)\pi(i-j)}{L}}{2k-1}\right)\left(\frac{(-1)^{i-j}}{2L\sin\frac{\pi(i-j)}{L}}\right).
\label{eq:twopoint}
\end{equation}
For every $d\in\{1,\ldots,L-1\}$, we consider a two-point correlator which is averaged over all pairs of spins at a distance $d$ on the ring. In particular, we let $M_d\equiv L^{-1}\sum_{i=1}^L Z_iZ_{\text{Mod}(i+d)}$ where $\text{Mod}(i+d)=((i+d-1)\text{ mod }L)+1$. Then the ground state satisfies
\begin{equation}
\bra{\psi}M_d\ket{\psi}=\bra{\psi}Z_1Z_{1+d}\ket{\psi}
\label{eq:md}
\end{equation}
and the exact value is therefore given by Eq.~\eqref{eq:twopoint}. We have tested our implementation of the CTMC by computing the expected value of $M_d$ for $d=1,5,10$. Since these observables are diagonal in the computational basis, they can be straightforwardly estimated from the output $\{\xi(\tau):0\leq \tau\leq T\}$ of Gillespie's algorithm. 
Here $T$ is the final time that determines the stopping condition of Gillespie's algorithm (see line 8). 
The time $T$ should not be confused with the number of flips, which is a random variable.
In particular, let $f:\{0,1\}^L\rightarrow\RR$ be such that $f(z)=\langle z|M_d|z\rangle$. Letting $\tau_0\leq T$ be an initial ``burn-in" time used to equilibrate the CTMC, we compute an approximation to $\Exp_{x\sim\pi}[f(x)]=\langle \psi|M_d|\psi\rangle$ using an estimator
\begin{equation}
\frac{1}{T-\tau_0}\int_{\tau=\tau_0}^{T} f(\xi(\tau))d\tau.
\label{eq:fest}
\end{equation}
(note that this integral can be equivalently be expressed as a finite sum since $\xi$ is piecewise constant). The data from our CTMC is shown in Fig.~\ref{fig:ctmctwopoint} and compared with the exact value computed using Eq.~\eqref{eq:twopoint}. See Appendix \ref{sec:details} for a description of the analysis used to estimate the error bars in this plot.

Next we compare our CTMC with the standard Metropolis-Hastings Markov Chain for sampling from the ground state probability distribution.
Fig.~\ref{eq:DTMC} shows the nearest-neighbor two-point correlator $M_1$ estimated using the MH method.  To use a Metropolis-Hastings Markov chain to sample from the ground state distribution $\pi$, we define the proposal distribution by viewing $H$ as the unweighted adjacency matrix of a transition graph with the state space $\{x\in\{0,1\}^L:|x|=\frac{L}{2}\}$. When the chain is in state $x$, it will propose to move to a neighbour $y$ of $x$ chosen uniformly at random. More explicitly, for every $x,y\in\{0,1\}^L$ such that $|x|=|y|=\frac{L}{2}$ and $x$ and $y$ differ in exactly two bits, the chain proposes to transition from $x$ to $y$ with probability $Q(y|x)=1/(L/2)^2=\frac{4}{L^2}$. The acceptance probabilities $A(y|x)=\min\{1,\frac{\pi(y)Q(x|y)}{\pi(x)Q(y|x)}\}$ are defined so that $Q$ and $A$ together satisfy the detailed balance condition. The fixed-node MH chain is defined by selecting a different proposal distribution, namely the one induced by viewing $F$ as an unweighted adjacency matrix.
More precisely, for a fixed state $x$ as above the proposal distribution $Q(y|x)$ is uniform
on the set of states $y$ with $|y|=\frac{L}{2}$ such that 
$\langle \psi|x\rangle \langle x|H|y\rangle \langle y|\psi\rangle \le 0$.
 We note that both MH chains are irreducible and aperiodic \footnote{Clearly the MH chain based on $H$ is irreducible. The one based on $F$ is also irreducible --otherwise $F$ would have a non-unique ground state, contradicting Lemma \ref{lemma:fixed_node}.   Both MH chains are aperiodic when $L\geq 4$. This can be seen to follow from the fact that they are irreducible and that there is at least one acceptance probability $A(y|x)$ which is less than one. The latter implies $\sum_y A(y|x)Q(y|x)<1$
and thus the self-loop probability for the state $x$
is strictly positive.} and therefore the distribution obtained by running the chain with any initial state converges to the limiting distribution $\pi$ as the number of steps grows.

\begin{figure}[H]
\centering
\begin{subfigure}{1\textwidth}
\centering
\includegraphics[scale=0.6]{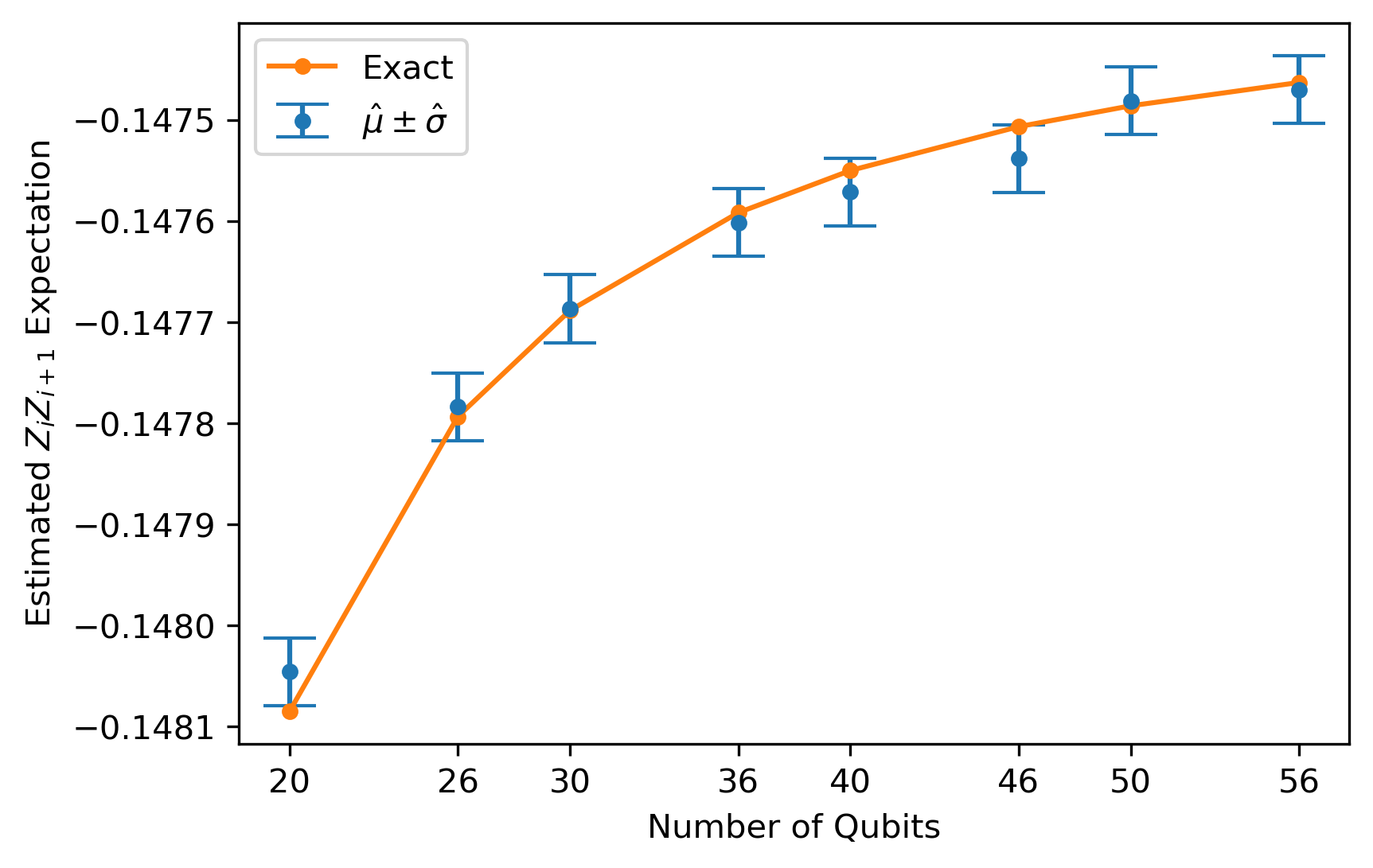}
\caption{}
\end{subfigure}

\vspace{0.1cm}
\begin{subfigure}{1\textwidth}
\centering
\includegraphics[scale=0.6]{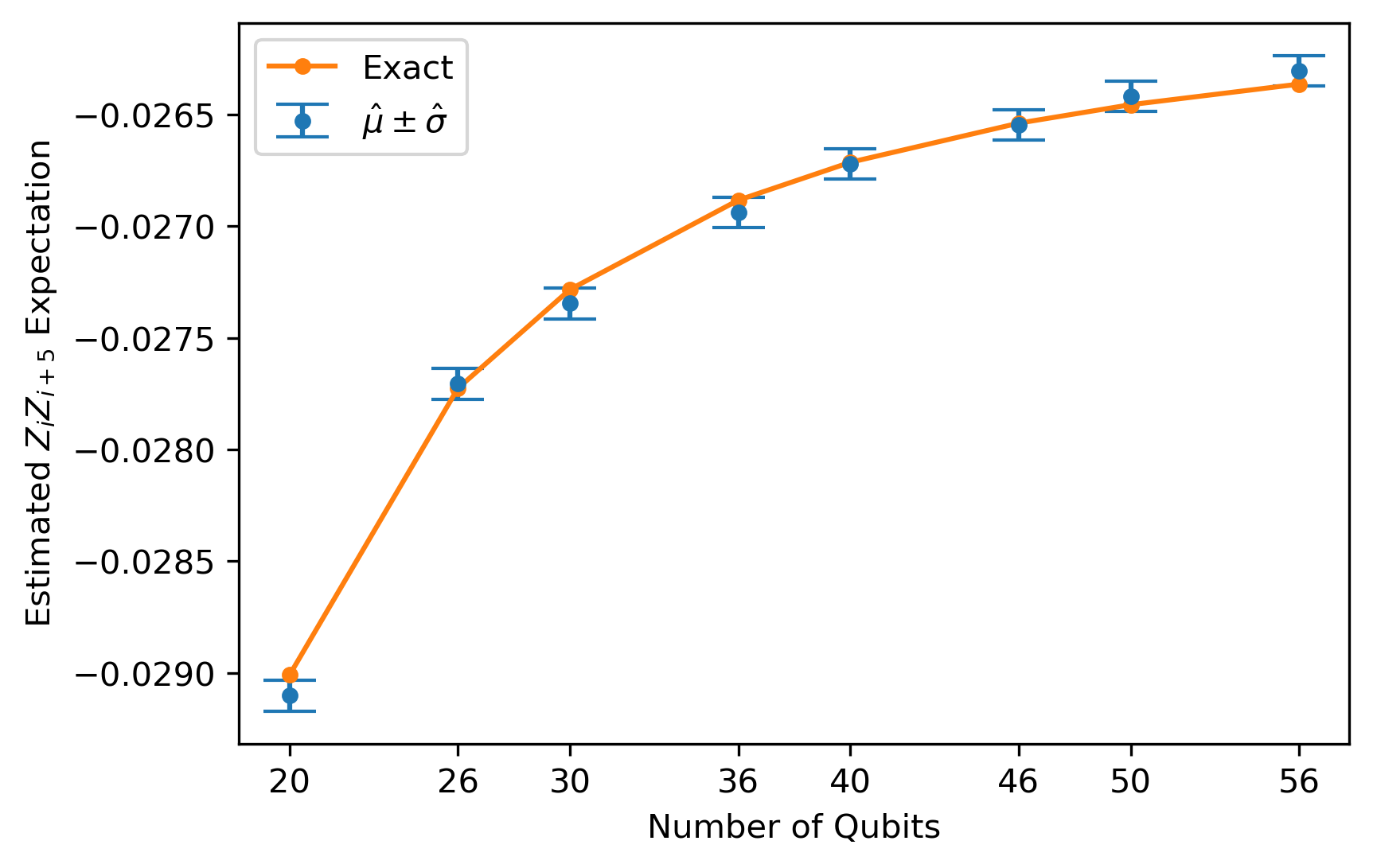}
\caption{}
\end{subfigure}

\vspace{0.1cm}
\begin{subfigure}{1\textwidth}
\centering
\includegraphics[scale=0.6]{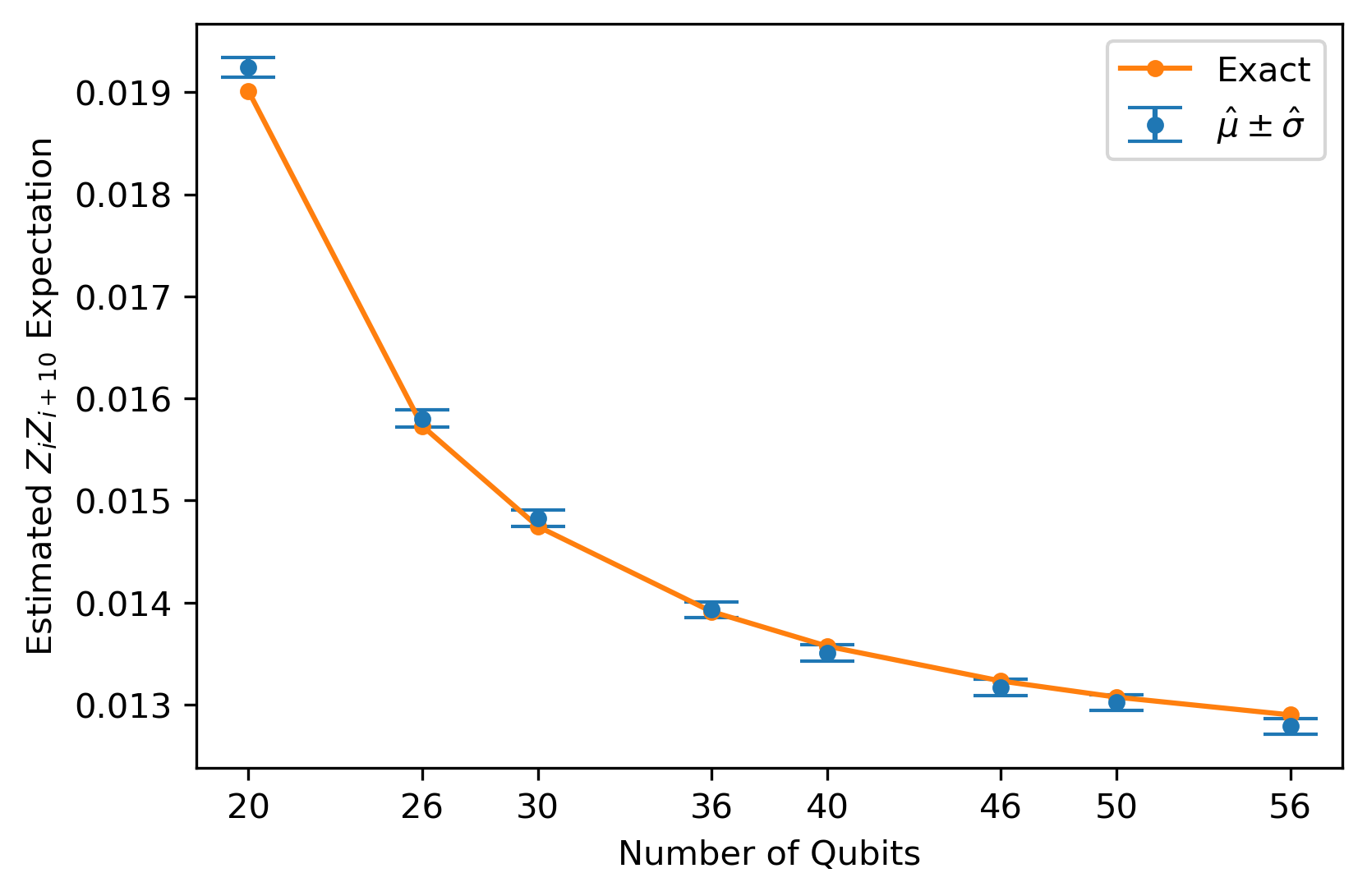}
\caption{}
\end{subfigure}
\caption{
Two-point correlation functions 
$\la \psi|Z_i Z_{i+d}|\psi\ra$ with
$d=1,5,10$
estimated using the continuous-time Markov Chain and compared with the exact formula Eq.~\eqref{eq:twopoint}. Here the CTMC was run for a total time $T=10^6$ and an initial ``burn-in" time $\tau_0=100$ was used for equilibration.
The estimator $\hat{\mu}$ is computed according to Eq.~\eqref{eq:fest} and the estimate of standard deviation
$\hat{\sigma}$ is constructed as described in 
Appendix~\ref{sec:details}.
\label{fig:ctmctwopoint}}
\end{figure}

\begin{figure}[H]
\begin{subfigure}{0.45\textwidth}
\includegraphics[scale=0.55]{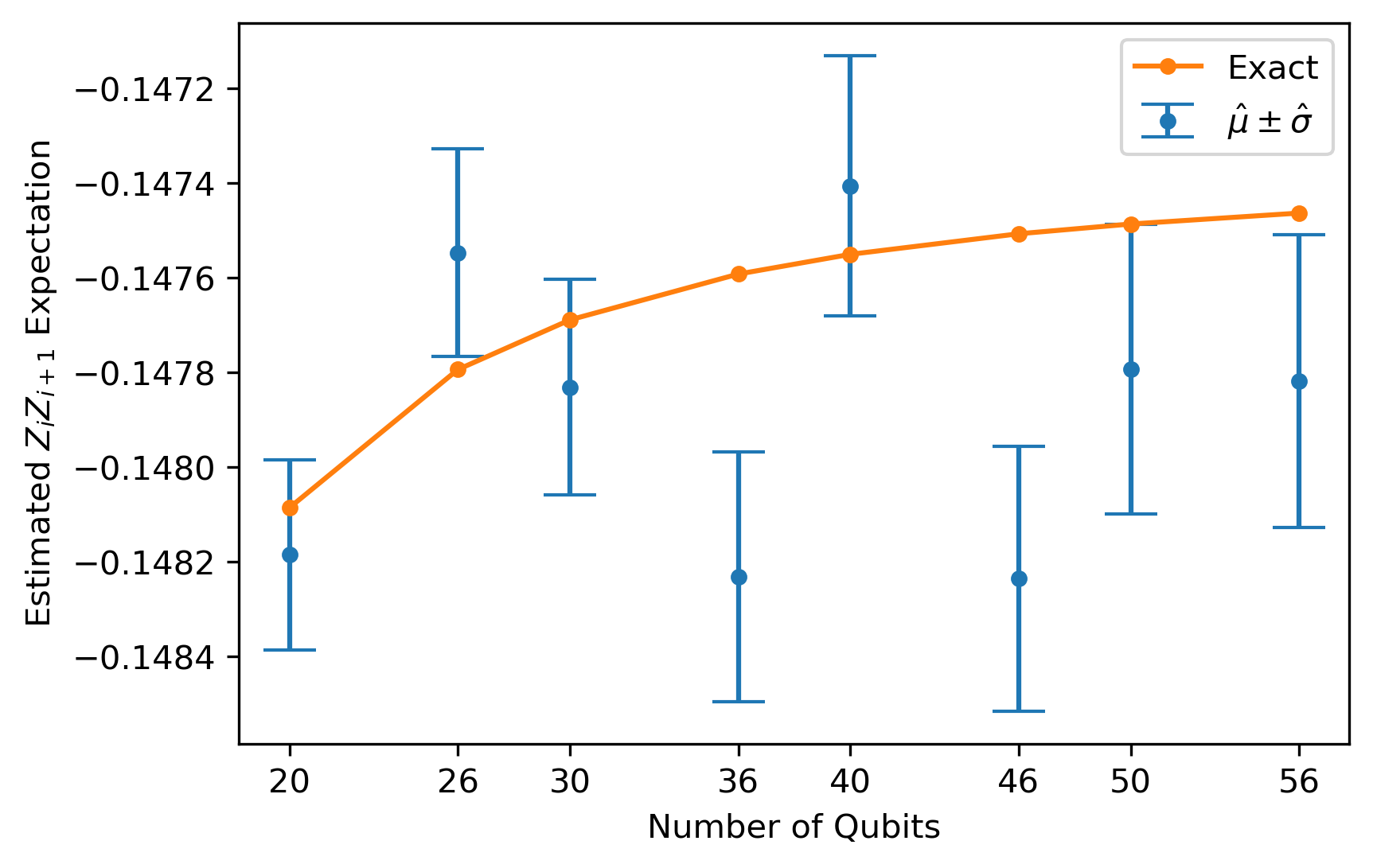}
\caption{}
\end{subfigure}
\hspace{1cm}
\begin{subfigure}{0.45\textwidth}
\includegraphics[scale=0.55]{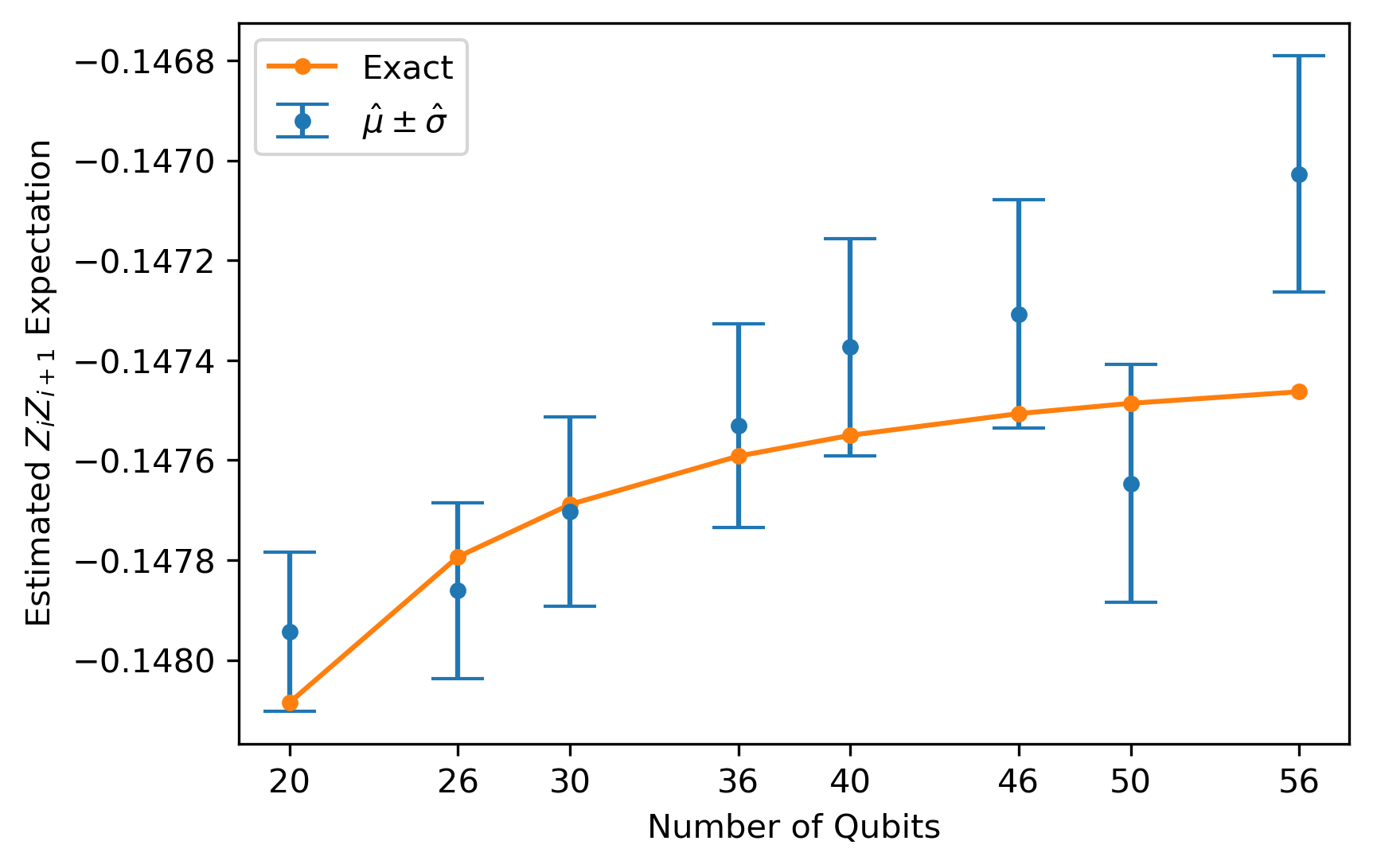}
\caption{}
\end{subfigure}
\caption{\label{eq:DTMC} Nearest-neighbor two-point correlator computed using a Metropolis-Hastings Markov Chain with $5\cdot 10^6$ steps of the Markov Chain  and $10^6$ used for equilibration. In (a) the Markov Chain proposal distribution was determined by $H$, and in (b) the proposal distribution was determined by $F$. }
\end{figure}

To enable a comparison between the Metropolis-Hastings Markov Chain and the CTMC we have tried to assess the computational cost of generating independent samples using each method.  For the MH chain, this is determined by the autocorrelation time as measured in the number of steps of the (discrete-time) Markov Chain. For the CTMC, we can also compute the autocorrelation time but the computational cost of running the chain for a given interval of time (using Gillespie's algorithm) is determined by the number of \textit{transitions} (i.e. flips) during the interval. For this reason we choose to normalize the autocorrelation time of the CTMC so that, roughly speaking, it is measured in units of transitions rather than time, see Appendix \ref{sec:details} for details.

\begin{figure}[H]
\centering
\begin{subfigure}{0.45\textwidth}
\includegraphics[scale=0.55]{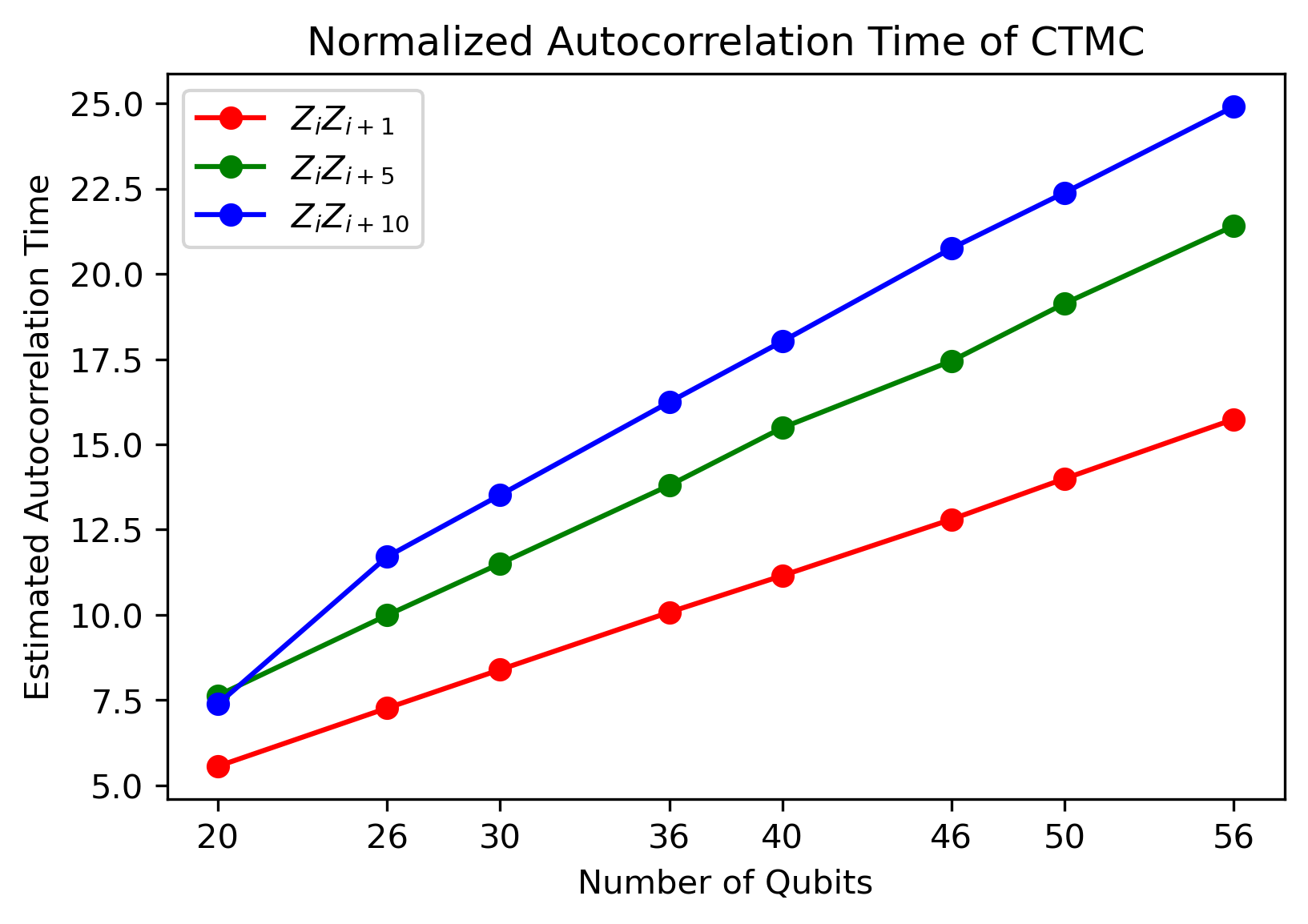}
\caption{}
\end{subfigure}
\hspace{1cm}
\begin{subfigure}{0.45\textwidth}
\includegraphics[scale=0.55]{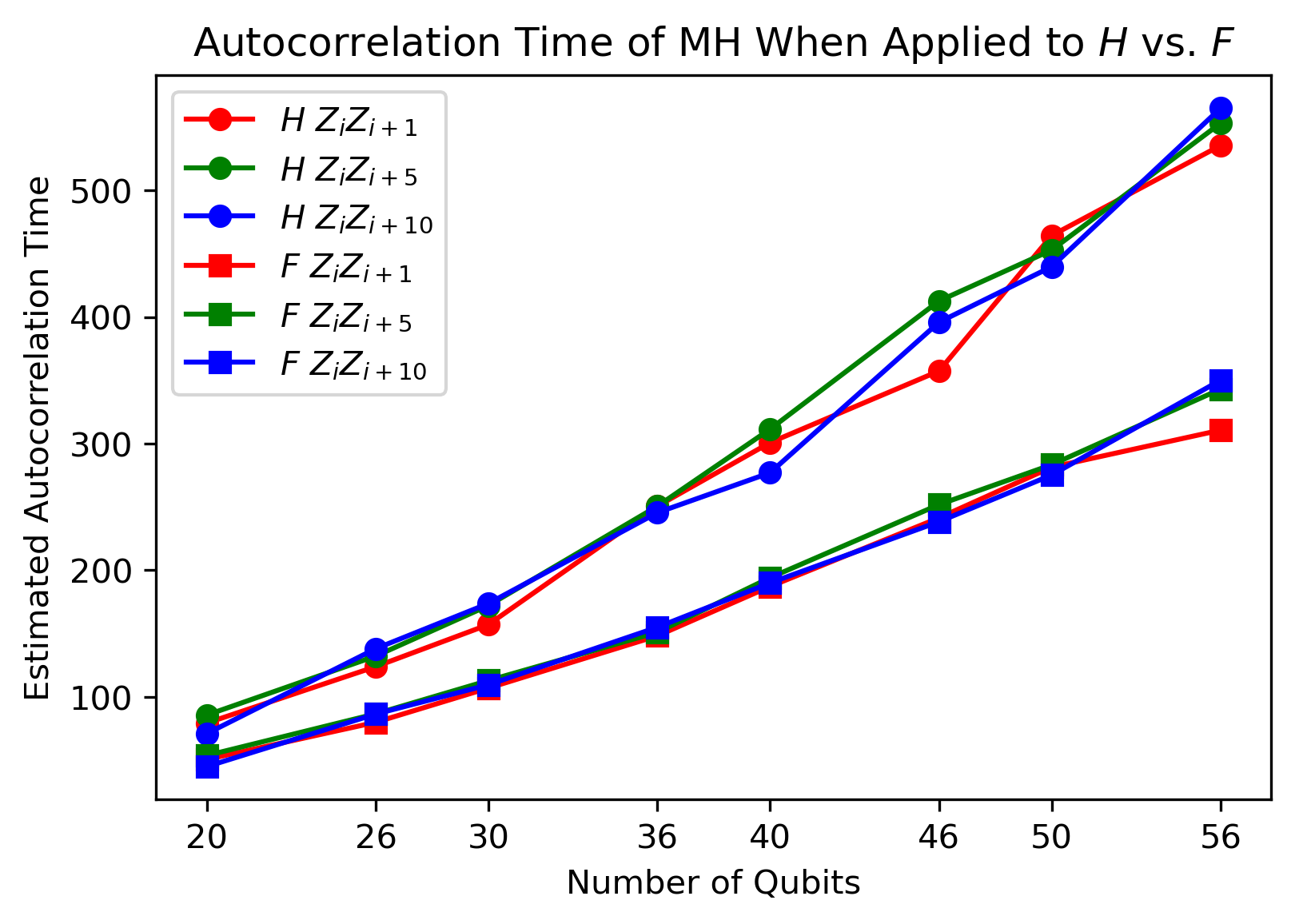}
\caption{}
\end{subfigure}
\caption{(a) Autocorrelation time of observables estimated using CTMC, measured in units of the number of transitions. (b) Autocorrelation time of observables estimated using Metropolis-Hastings with proposal distributions generated by either the Haldane-Shastry Hamiltonian $H$ or the associated fixed-node Hamiltonian $F$.\label{fig:auto}}
\end{figure}

We note that the normalized autocorrelation times in Fig.~\ref{fig:auto} (a) are significantly lower than the ones for the Metropolis-Hastings algorithm reported in Fig.~\ref{fig:auto} (b). Intriguingly, we also observe in Fig.~\ref{fig:auto} (b) that the Metropolis-Hastings algorithm using the fixed-node Hamiltonian $F$ has shorter autocorrelation times for each number of qubits and each of the three observables. In other words, using the fixed-node Hamiltonian to generate the proposal distribution appears to have a marked advantage over the naive strategy based on the Hamiltonian itself.  This is despite the fact that the nonzero matrix elements of $F$ are a subset of those of $H$. It appears that the information about the sign structure of $\psi$ that determines which entries of $F$ are set to zero may help the MH chain converge more quickly.

To investigate this further, we looked at how the autocorrelation times behave when we use a corrupted ground state distribution $\tilde{\pi}$ as opposed to the true ground state distribution $\pi$. For some noise strength $\kappa>0$, we construct a perturbed distribution $\tilde{\pi}$ by defining
\begin{equation*}
\braket{x}{\tilde{\psi}}=\braket{x}{\psi}+N\left(0,\frac{\kappa}{2^L}\right)
\end{equation*}
and setting $\tilde{\pi}(x)=|\braket{x}{\tilde{\psi}}|^2$ for every $x\in\{0,1\}^L$ where $N(\mu,\sigma)$ is a Gaussian random variable with mean $\mu$ and standard deviation $\sigma$. $\tilde{\pi}$ is normalized so that $\sum_{x\in\{0,1\}^L}\tilde{\pi}(x)=1$. Fig.~\ref{fig:corrupted} shows
that as the total variance distance between $\pi$ and $\tilde{\pi}$ increases, the differences between the autocorrelation times based on $H$ (Haldane-Shastry Hamiltonian) and $F$ (defined by Eqs.~(\ref{eq:fnode1},\ref{eq:fnode2},\ref{fixed_node_def}) with $\psi$ replaced by $\tilde{\psi}$)  disappear as expected.

\begin{figure}[H]
\centering
\includegraphics[scale=0.65]{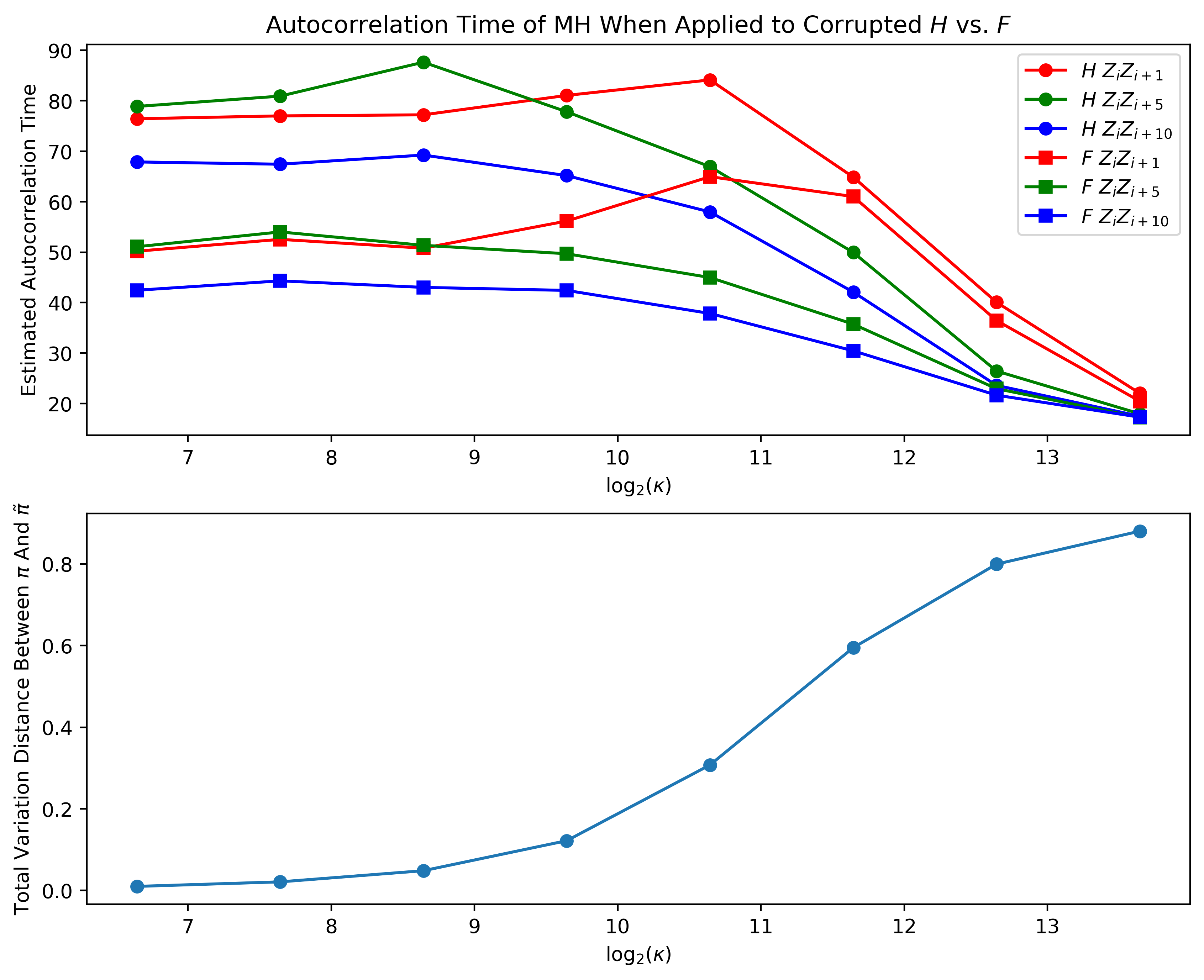}
\caption{Autocorrelation times of Metropolis-Hastings chains defined using corrupted ground state distributions. Each datapoint is computed using a length 1000000 chain with a 1000-sample burn in period, all for $L=20$. When the noise strength $\kappa$ is small, applying the fixed-node transformation reduces autocorrelation times for each of the chosen observables. As the error strength increases (as the corrupted distributions become dominated by noise), the differences between the autocorrelation times of F and H vanish.}
\label{fig:corrupted}
\end{figure}

\section{Acknowledgments}
SB thanks Vojtech Havlicek for helpful discussions. 
DG and YL acknowledge the support of the
Natural Sciences and Engineering Research Council of
Canada through grant number RGPIN-2019-04198. DG
also acknowledges the support of the Canadian Institute
for Advanced Research, and IBM Research. Research at
Perimeter Institute is supported in part by the Government of Canada through the Department of Innovation,
Science and Economic Development Canada and by the
Province of Ontario through the Ministry of Colleges and
Universities.
\appendix

\section{Haldane-Shastry state: lack of free fermion representation}
\label{sec:lack}
Given a bit string $x\in \{0,1\}^n$, let $|x|$ be the number of 1s in $x$ (the Hamming weight).
Suppose $n$ is an even integer.
Consider an $n$-qubit state $|\psi_n\ra$ with amplitudes
\be
\label{HSeq1}
\la x|\psi_n\ra = e^{i\theta_x} \prod_{1\le a<b\le n} \left[ \sin{\left( \frac{ \pi(a-b)}n \right)}\right]^{2x_a x_b}
\ee
if $|x|=n/2$ and $\la x|\psi_n\ra=0$ otherwise. Here $\theta_x\in \RR$ are arbitrary angles.
For example,
\be
\label{example4}
|\psi_4\ra = \frac{e^{i\theta_{1100}}} 2 |1100\ra + \frac{e^{i\theta_{0011}}}2 |0011\ra + \frac{e^{i\theta_{1001}}}2 |1001\ra +  \frac{e^{i\theta_{0110}}} 2 |0110\ra + e^{i\theta_{1010}}|1010\ra + e^{i\theta_{0101}} |0101\ra.
\ee
The state defined in Eq.~(\ref{HSeq1}) coincides with the Haldane-Shastry state 
for a suitable choice of the phase factors $e^{i\theta_x} = \pm 1$.

Define Majorana fermion operators $\gamma_1,\gamma_2,\ldots,\gamma_{2n}$ such that 
$\gamma_1=X_1$, $\gamma_2=Y_1$, 
\be
\gamma_{2p-1} = Z_1\cdots Z_{p-1} X_p \quad \mbox{and} \quad \gamma_{2p} = Z_1 \cdots Z_{p-1} Y_p
\ee
for $p=2,\ldots,n$. 
Recall that $n$-qubit state $|\phi\ra$ is called a free fermion state~\cite{terhal2002classical} if it obeys fermionic Wick's theorem, that is,
for any $k$-tuple of Majorana operators $1\le p_1<p_2<\ldots<p_k\le 2n$ one has
\be
\la \phi |\gamma_{p_1} \gamma_{p_2} \cdots \gamma_{p_{k}} |\phi\ra =\left\{
\ba{rcl}
0 &\mbox{if} & \mbox{$k$ is odd} \\
\calA\left(
\la \phi| \gamma_{p_1} \gamma_{p_2} |\phi\ra \cdot
\la \phi| \gamma_{p_3} \gamma_{p_4} |\phi\ra \cdots
  \la \phi| \gamma_{p_{k-1}} \gamma_{p_{k}}|\phi\ra\right) &\mbox{if} & \mbox{$k$ is even} \\
\ea\right.
\ee
Here $\calA$ denotes anti-symmetrization over all permutations of indices $p_1,p_2,\ldots,p_{k}$. 
For example,
\[
\la \phi|\gamma_1 \gamma_2 \gamma_3 \gamma_4|\phi\ra = \la \phi|\gamma_1\gamma_2|\phi\ra \cdot
\la \phi|\gamma_3\gamma_4|\phi\ra  - \la \phi|\gamma_1\gamma_3|\phi\ra \cdot
\la \phi|\gamma_2\gamma_4|\phi\ra + \la \phi|\gamma_1\gamma_4|\phi\ra \cdot
\la \phi|\gamma_2\gamma_3|\phi\ra.
\]
The task of sampling  the probability distribution $|\la x|\phi\ra|^2$ with a free fermion state $|\phi\ra$
admits an efficient classical algorithm with the runtime $O(n^3)$, see for instance~\cite{terhal2002classical}
or Appendix~B of~\cite{bravyi2018correcting}.
However, these free fermion simulation methods are not directly applicable to 
simulating measurement of the Haldane-Shastry state, as follows from the following lemma.
\begin{lemma}
Let $|\psi_n\ra$ be the state defined in Eq.~(\ref{HSeq1}) with an arbitrary choice of the angles $\theta_x$.
Suppose $n\ge 4$ is an integer multiple of four. Then $|\psi_n\ra$ is not proportional to a free fermion state.
\end{lemma}
\begin{proof}
We shall use the following
well-known facts.
\begin{fact}[\cite{bravyi2009contraction}]
\label{fact:fact1}
Suppose  a four-qubit state $|\phi\ra$ has support only on even-weight basis vectors.
Then $|\phi\ra$ is a free fermion state if and only if
\be
\label{ff4}
-\la 0000|\phi\ra \cdot \la 1111|\phi\ra + 
\la 1100|\phi\ra \cdot \la 0011|\phi\ra - \la 1010|\phi\ra \cdot \la 0101|\phi\ra + \la 1001|\phi\ra \cdot \la 0110|\phi\ra = 0.
\ee
\end{fact}
\begin{fact}
\label{fact:fact2}
Suppose $|\phi\ra$ is an $n$-qubit free fermion state and $D$ is a tensor product of diagonal single-qubit operators.
Then  $D|\phi\ra$ is proportional to a free fermion state.
\end{fact}
\begin{proof}
Suppose $D$ acts non-trivially on a single qubit $j$. Then $D$ is a linear combination of the identity $I$
and $Z_j=-i\gamma_{2j-1}\gamma_{2j}$.
            Thus $D\sim e^{\alpha \gamma_{2j-1}\gamma_{2j}}$ for some complex number $\alpha$.
In the general case $D$ is a product of operators as above. Thus we can write $D\sim e^{\Gamma}$,
where $\Gamma$ is some operator  quadratic in $\gamma_1,\ldots,\gamma_{2n}$.            
It is well known that matrix exponentials $e^{\Gamma}$ with a quadriatic fermionic operator $\Gamma$ map
free fermion states to free fermion states up to the normalization~\cite{flo}. 
\end{proof}
\begin{fact}[\cite{terhal2002classical,flo}]
\label{fact:fact3}
Any computational basis state $|x\ra$ is a free fermion state. 
A tensor product $\phi_1\otimes \phi_2$ is a free fermion state
iff $\phi_1$ and $\phi_2$ are free fermion states.
\end{fact}
Assume that  $|\psi_4\ra\sim |\phi\ra$ for some free fermion state $|\phi\ra$.
Then Eq.~(\ref{ff4}) gives
\be
|\la 1010|\psi_4\ra \cdot \la 0101|\psi_4\ra| \le | \la 1100|\psi_4\ra \cdot \la 0011|\psi_4\ra | + |  \la 1001|\psi_4\ra \cdot \la 0110|\psi_4\ra| 
\ee 
which contradicts to Eq.~(\ref{example4}). This proves the lemma for $n=4$.
Consider an  integer $n=4m\ge 4$.
Partition the set of $n$ qubits as $[n]=AB$ where 
\[
A=\{1,m+1,2m+1,3m+1\}
\]
and $B$ is the complement of $A$.
Given bit strings $x^A\in \{0,1\}^{|A|}$ and $x^B\in \{0,1\}^{|B|}$, let
$x^A x^B\in \{0,1\}^n$ be a string whose projection onto $A$ and $B$
coincides with $x^A$ and $x^B$ respectively. 
Suppose $|x^A|=2$ and $|x^B|=n/2 - 2$.
Then Eq.~(\ref{HSeq1}) implies
\be
\label{ABeq1}
\la x^A x^B|\psi_n\ra = e^{i\theta(x^A,x^B)}  \prod_{j=1}^4 f_j(x^A_j,x^B) \prod_{1\le q<r\le 4}  
\left[ \sin{\left( \frac{ \pi(q-r)}4 \right)}\right]^{2x^A_q x^A_r}
\ee
for some real-valued functions  $\theta(x^A,x^B)$ and $f_j(x^A_j,x^B)$.
Note that $\la x|\psi_n\ra\ne 0$ whenever $|x|=n/2$. Thus $f_j(x^A_j,x^B)\ne 0$
for all $x^A$ and $x^B$ as above.
From Eq.~(\ref{ABeq1}) one gets
\be
\label{ABeq2}
\la x^A x^B |\psi_n\ra = \la x^A | D_1 D_2 D_3 D_4|\psi_4\ra,
\ee
where  $|\psi_4\ra$ is defined in Eq.~(\ref{example4}) and $D_j$ are diagonal invertible single-qubit operators
such that $D_j|x^A\ra = f_j(x^A_j,x^B) |x^A\ra$ for all $x^A,x^B$.
Assume that $|\psi_n\ra$ is free. Then Fact~\ref{fact:fact2} implies that
$(I_A \otimes |x^B\ra\la x^B|)|\psi_n\ra$ is proportional to a free state.
Write
\[
(I_A \otimes |x^B\ra\la x^B)|\psi_n\ra  \sim |\psi^A(x^B)\ra \otimes |x^B\ra
\]
for some normalized four-qubit state $|\psi^A(x^B)\ra$.
Fact~\ref{fact:fact3} implies that  $|\psi^A(x^B)\ra$ is free.
From Eq.~(\ref{ABeq2}) one gets
\[
|\psi_4\ra \sim D_1^{-1} D_2^{-1} D_3^{-1} D_4^{-1} |\psi^A(x^B)\ra
\]
Using  Fact~\ref{fact:fact2} again one infers that $|\psi_4\ra$ is free.
However this is a contradiction since we have already proved that $|\psi_4\ra$ is not free.
Thus $|\psi_n\ra$ is not proportional to a free state.
\end{proof}

\section{Sign problem in Haldane-Shastry Hamiltonian}
\label{sec:lack1}

Recall that an $n$-qubit Hamiltonian $H$ is called  stoquastic (sign problem free) if $H$ has real
matrix elements in the standard basis and $\la x|H|y\ra\le 0$ for all $x\ne y$. 
Consider a Hamiltonian 
\be
\label{XYZ}
H = \sum_{1\le i<j\le n} J_{i,j} (X_i X_j + Y_i Y_j + Z_i Z_j),
\ee
where $J_{i,j}>0$ are arbitrary coefficients. This includes Haldane-Shastry Hamiltonian Eq.~(\ref{eq:hsh})
as a special case. Choose any basis vectors $x,y\in \{0,1\}^n$ that differ only on two qubits $i,j$
such that $x_ix_j=10$, $y_iy_j=01$, and $x_\ell=y_\ell$ for all $\ell\notin \{i,j\}$. 
A simple calculation gives
$\la x|H|y\ra = 2J_{i,j}>0$, that is, $H$ is not stoquastic. Suppose  there exist single-qubit unitary operators $U_1,U_2,\ldots,U_n$ 
such that 
\[
H_U := (U_1\otimes U_2 \otimes \cdots \otimes U_n)^\dag H (U_1\otimes U_2 \otimes \cdots \otimes U_n)
\]
is stoquastic. The question of whether Hamiltonians of the form Eq.~(\ref{XYZ}) can be
made stoquastic by a local change of basis has been studied by Klassen and Terhal~\cite{klassen2019two}.
Lemma~22 of Ref.~\cite{klassen2019two} implies that without loss of generality the unitaries $U_j$
can be chosen from a {\em finite}  subgroup of the unitary group (known as the Clifford group). 
Hence the number of {\em distinct} unitaries $U_j$ is upper bounded by a constant independent of 
$n$. Thus
 for a sufficiently large number of qubits $n$, there will be at least one pair of qubits $i< j$
such that $U_i=U_j$. Note that  $U_i\otimes U_j$ commutes with $X\otimes X + Y\otimes Y + Z\otimes Z$
if $U_i=U_j$.
Thus we can write 
\[
H_U = J_{i,j}  (X_i X_j + Y_i Y_j + Z_i Z_j) + H_{\mathrm{else}},
\]
where $H_{\mathrm{else}}$ is a sum of operators that act non-trivially on at most one of the qubits $i,j$.
Now the same calculation as above shows that 
$\la x|H_U|y\ra = 2J_{i,j}>0$, that is, $H_U$ is not stoquastic. 

\section{Details of numerical implementation\label{sec:details}}

In this appendix we describe how error bars are computed in Fig.~\ref{fig:ctmctwopoint}. We also describe the definition of the normalized autocorrelation time reported in Fig.~\ref{fig:auto}.

Let $\{\xi(\tau):\tau\geq 0\}$ denote the stochastic process induced by running Gillespie's algorithm on the HS fixed-node Hamiltonian where the initial state $\xi(0)\sim\pi$, the true ground state distribution. The following derivation is based on Ref.~\cite{kennedy2016notes}. Let $f:\{0,1\}^L\rightarrow\RR$ be a function and suppose our goal is to estimate the mean of $f$ with respect to the steady distribution $\pi$. Let $T>0$ be large and $h>0$ be small. We can estimate $\mu=\Exp_{x\sim\pi}[f(x)]$ using the estimator
$$\hat{\mu}=\frac{1}{\lfloor\frac{T}{h}\rfloor}\sum_{j=1}^{\lfloor\frac{T}{h}\rfloor} f(\xi(jh)).$$
Note that in the limit $h\rightarrow 0$ this can be represented as an integral (cf. Eq.~\eqref{eq:fest}). For the purposes of estimating error bars it will be convenient to use the discretized representation however.  Since $\xi(0)\sim\pi$, $\xi(\tau)\sim\pi$ for every $\tau\geq 0$. Let $\sigma^2_f=\Var(f(\xi(0)))=\Var(f(\xi(\tau)))$ for every $\tau\geq 0$. Recall that the Pearson correlation coefficient of two random variables $A$ and $B$ is defined by
$$\cor(A,B)=\frac{\cov(A,B)}{\sigma_A\sigma_B}$$
where $\sigma_A^2$ and $\sigma_B^2$ are the variances of $A$ and $B$ respectively. Then, the variance of the estimator $\hat{\mu}$ is
\begin{align*}
\hat{\sigma}^2=\Var(\hat{\mu})&=\frac{1}{\lfloor\frac{T}{h}\rfloor^2}\Var\left(\sum_{j=1}^{\lfloor\frac{T}{h}\rfloor} f(\xi(jh))\right)\\
&=\frac{1}{\lfloor\frac{T}{h}\rfloor^2}\sum_{i=1}^{\lfloor\frac{T}{h}\rfloor}\sum_{j=1}^{\lfloor\frac{T}{h}\rfloor} \cov(f(\xi(ih)),f(\xi(jh)))\\
&=\frac{\sigma^2_f}{\lfloor\frac{T}{h}\rfloor^2}\sum_{i=1}^{\lfloor\frac{T}{h}\rfloor}\sum_{j=1}^{\lfloor\frac{T}{h}\rfloor}\cor(f(\xi(ih)),f(\xi(jh))).
\end{align*}
We assume for every $i\in\{1,\ldots,\lfloor\frac{T}{h}\rfloor\}$, each correlation term in $\tau_f=\sum_{j=1}^{\lfloor\frac{T}{h}\rfloor}\cor(f(\xi(ih)),f(\xi(jh)))$ only depends on $|i-j|$ and is otherwise independent of $i$. Thus,
$$\hat{\sigma}^2=\frac{\sigma^2_f\tau_f}{\lfloor\frac{T}{h}\rfloor}.$$
$\tau_f$ is the integrated autocorrelation time w.r.t $f$. We use the emcee library \cite{foreman2013emcee} to obtain an estimate $\hat{\tau}_f$ of $\tau_f$. We further estimate $\hat{\sigma}^2$ using $\hat{\tau}_f$ and the sample variance of $\{f(\xi(1h)),f(\xi(2h)),\ldots,f(\xi(\lfloor\frac{T}{h}\rfloor h))\}$.

Notice that $\hat{\tau}_fh$ estimates the autocorrelation \textit{time} of the CTMC. Let $r$ denote the total number of transitions it took for the CTMC to reach \textit{time} $T$. Then $\frac{r}{T}$ gives the average number of transitions needed for the CTMC to advance \textit{time} by $1$ unit. Thus, we infer that it takes an average of $\hat{\tau}_fh\frac{r}{T}$ transitions for the CTMC to advance \textit{time} by $\hat{\tau}_fh$ units. In Figure \ref{fig:auto} we plot the normalized autocorrelation time $\hat{\tau}_fh\frac{r}{T}$ against the number of qubits.

\bibliographystyle{unsrt}
\bibliography{mybib}

\begin{thebibliography}{10}

\bibitem{barahona1982computational}
Francisco Barahona.
\newblock On the computational complexity of {I}sing spin glass models.
\newblock {\em \href{https://doi.org/10.1088/0305-4470/15/10/028}{Journal of
  Physics A: Mathematical and General}}, 15(10):3241, 1982.

\bibitem{long2010restricted}
Philip~M Long and Rocco~A Servedio.
\newblock Restricted {B}oltzmann machines are hard to approximately evaluate or
  simulate.
\newblock {\em
  \href{https://dl.acm.org/doi/abs/10.5555/3104322.3104412}{Proceedings of the
  27th International Conference on International Conference on Machine
  Learning. ICML’10. Haifa, Israel}}, page 703–710, 2010.

\bibitem{hastings_monte_1970}
W.~K. Hastings.
\newblock Monte {Carlo} sampling methods using {Markov} chains and their
  applications.
\newblock {\em \href{https://doi.org/10.2307/2334940}{Biometrika}},
  57(1):97--109, April 1970.

\bibitem{levin2017markov}
David~A Levin and Yuval Peres.
\newblock {\em Markov chains and mixing times}, volume 107.
\newblock \href{https://doi.org/10.1090/mbk/107 }{American Mathematical Soc.},
  2017.

\bibitem{bravyi2021simulate}
Sergey Bravyi, David Gosset, and Yinchen Liu.
\newblock How to simulate quantum measurement without computing marginals.
\newblock {\em \href{https://doi.org/10.1103/PhysRevLett.128.220503}{Physical
  Review Letters}}, 128(22):220503, 2022.

\bibitem{aharonov2003adiabatic}
Dorit Aharonov and Amnon Ta-Shma.
\newblock Adiabatic quantum state generation and statistical zero knowledge.
\newblock In {\em \href{https://doi.org/10.1145/780542.780546 }{Proceedings of
  the thirty-fifth annual ACM symposium on Theory of computing}}, pages 20--29,
  2003.

\bibitem{bravyiterhal}
Sergey Bravyi and Barbara Terhal.
\newblock Complexity of stoquastic frustration-free hamiltonians.
\newblock {\em \href{https://doi.org/10.1137/08072689X}{{SIAM} {J}ournal on
  {C}omputing}}, 39(4):1462--1485, 2010.

\bibitem{ten1995proof}
DFB Ten~Haaf, HJM Van~Bemmel, JMJ Van~Leeuwen, W~Van~Saarloos, and DM~Ceperley.
\newblock Proof for an upper bound in fixed-node {M}onte {C}arlo for lattice
  fermions.
\newblock {\em \href{https://doi.org/10.1103/physrevb.51.13039 }{Physical
  Review B}}, 51(19):13039, 1995.

\bibitem{foulkes2001quantum}
WMC Foulkes, Lubos Mitas, RJ~Needs, and Guna Rajagopal.
\newblock Quantum monte carlo simulations of solids.
\newblock {\em \href{https://doi.org/10.1103/RevModPhys.73.33}{Reviews of
  Modern Physics}}, 73(1):33, 2001.

\bibitem{becca2017quantum}
Federico Becca and Sandro Sorella.
\newblock {\em Quantum Monte Carlo Approaches for Correlated Systems}.
\newblock \href{https://doi.org/10.1017/9781316417041}{Cambridge University
  Press}, 2017.

\bibitem{havlicek2022amplitude}
Vojtech Havlicek.
\newblock Amplitude ratios and neural network quantum states.
\newblock {\em \href{https://doi.org/10.22331/q-2023-03-02-938}{Quantum}},
  7:938, 2023.

\bibitem{gillespie1977exact}
Daniel~T Gillespie.
\newblock Exact stochastic simulation of coupled chemical reactions.
\newblock {\em \href{https://doi.org/10.1021/j100540a008 }{The journal of
  physical chemistry}}, 81(25):2340--2361, 1977.

\bibitem{diaconis1991geometric}
Persi Diaconis and Daniel Stroock.
\newblock Geometric bounds for eigenvalues of {M}arkov chains.
\newblock {\em \href{https://doi.org/10.1214/aoap/1177005980 }{The Annals of
  Applied Probability}}, pages 36--61, 1991.

\bibitem{takahara2017notes}
Glen Takahara.
\newblock {\em STAT 455 Stochastic Process Lecture Notes}.
\newblock 2017.

\bibitem{prokof1998exact}
NV~Prokof’Ev, BV~Svistunov, and IS~Tupitsyn.
\newblock Exact, complete, and universal continuous-time worldline monte carlo
  approach to the statistics of discrete quantum systems.
\newblock {\em \href{https://doi.org/10.1134/1.558661}{Journal of Experimental
  and Theoretical Physics}}, 87(2):310--321, 1998.

\bibitem{smallgaps}
Edward Farhi, Jeffrey Goldstone, David Gosset, Sam Gutmann, Harvey~B. Meyer,
  and Peter Shor.
\newblock \href{https://doi.org/10.26421/qic11.3-4-1}{Quantum Adiabatic
  Algorithms, Small Gaps, and Different Paths}.
\newblock {\em Quantum Info. Comput.}, 11(3):181–214, mar 2011.

\bibitem{stephan2017full}
Jean-Marie Stephan and Frank Pollmann.
\newblock Full counting statistics in the haldane-shastry chain.
\newblock {\em \href{https://doi.org/10.1103/physrevb.95.035119}{Physical
  Review B}}, 95(3):035119, 2017.

\bibitem{pai2020disordered}
Shriya Pai, NS~Srivatsa, and Anne~EB Nielsen.
\newblock Disordered haldane-shastry model.
\newblock {\em \href{https://doi.org/10.1103/physrevb.102.035117}{Physical
  Review B}}, 102(3):035117, 2020.

\bibitem{klassen2019two}
Joel Klassen and Barbara~M Terhal.
\newblock Two-local qubit hamiltonians: when are they stoquastic?
\newblock {\em \href{ https://doi.org/10.22331/q-2019-05-06-139}{Quantum}},
  3:139, 2019.

\bibitem{nielsen2012laughlin}
Anne~EB Nielsen, J~Ignacio Cirac, and Germ{\'a}n Sierra.
\newblock Laughlin spin-liquid states on lattices obtained from conformal field
  theory.
\newblock {\em \href{https://doi.org/10.1103/PhysRevLett.108.257206}{Physical
  review letters}}, 108(25):257206, 2012.

\bibitem{vehtari2021rank}
Aki Vehtari, Andrew Gelman, Daniel Simpson, Bob Carpenter, and Paul-Christian
  B{\"u}rkner.
\newblock Rank-normalization, folding, and localization: An improved r for
  assessing convergence of mcmc (with discussion).
\newblock {\em \href{https://doi.org/10.1214/20-ba1221}{Bayesian analysis}},
  16(2):667--718, 2021.

\bibitem{terhal2002classical}
Barbara~M Terhal and David~P DiVincenzo.
\newblock Classical simulation of noninteracting-fermion quantum circuits.
\newblock {\em \href{https://doi.org/10.1103/physreva.65.032325}{Physical
  Review A}}, 65(3):032325, 2002.

\bibitem{bravyi2018correcting}
Sergey Bravyi, Matthias Englbrecht, Robert K{\"o}nig, and Nolan Peard.
\newblock Correcting coherent errors with surface codes.
\newblock {\em \href{https://doi.org/10.1038/s41534-018-0106-y }{npj Quantum
  Information}}, 4(1):1--6, 2018.

\bibitem{bravyi2009contraction}
Sergey Bravyi.
\newblock Contraction of matchgate tensor networks on non-planar graphs.
\newblock {\em \href{https://doi.org/10.1090/conm/482/09419 }{Contemp. Math}},
  482:179--211, 2009.

\bibitem{flo}
Sergey Bravyi.
\newblock Lagrangian representation for fermionic linear optics.
\newblock {\em \href{https://doi.org/10.26421/qic5.3-3 }{Quantum Information \&
  Computation}}, 5(3):216--238, 2005.

\bibitem{kennedy2016notes}
Tom Kennedy.
\newblock {\em Monte Carlo Methods - a special topics course}.
\newblock 2016.

\bibitem{foreman2013emcee}
Daniel Foreman-Mackey, David~W Hogg, Dustin Lang, and Jonathan Goodman.
\newblock emcee: the mcmc hammer.
\newblock {\em \href{https://doi.org/10.1086/670067 }{Publications of the
  Astronomical Society of the Pacific}}, 125(925):306, 2013.

\end{thebibliography}

\end{document}